%% file: new-main.tex
\documentclass[10pt]{article}
\usepackage{fullpage,amsfonts,amsmath,amsthm,amssymb,mathtools,mathrsfs,graphicx,upgreek}
\usepackage[left=2.54cm,top=2.54cm,right=2.54cm,bottom=2.54cm]{geometry}

\usepackage{algorithm}
\usepackage[noend]{algpseudocode}

\usepackage{epsfig}
\usepackage{bm}
\usepackage{comment}

\usepackage{hyperref}

\usepackage{enumitem}
\newcommand{\subscript}[2]{$#1 _ #2$}
\numberwithin{equation}{section}

\newtheorem{theorem}{Theorem}[section]
\newtheorem{proposition}[theorem]{Proposition}
\newtheorem{corollary}[theorem]{Corollary}

\newtheorem{lemma}[theorem]{Lemma}

\theoremstyle{definition}
\newtheorem{definition}{Definition}

\newtheorem{remark}[theorem]{Remark}

\usepackage{bbm}

\usepackage{algpseudocode}
\usepackage{datetime}

\newcommand{\st}{\text{st}}
\newcommand{\sta}{\text{st}}
\newcommand{\scr}{\mathcal}
\newcommand{\mb}{\mathbb}
\newcommand{\til}{\widetilde}

\newcommand{\eps}{\varepsilon}

\newcommand{\val}{\text{val}}
\newcommand{\OPT}{\text{OPT}}
\newcommand{\nOPT}{\text{OPT}_\text{n-adp}}
\newcommand{\rOPT}{\text{OPT}_{\text{rel}}}
\newcommand{\LPOPT}{\text{LPOPT}}
\newcommand{\poly}{\text{poly}}
\newcommand{\typ}{\text{typ}}
\newcommand{\qLPOPT}{\text{LPOPT}_{\text{QC}}}

\newcommand{\Ber}{\textup{Ber}}

\usepackage[usenames,dvipsnames]{xcolor}

\begin{document}
\title{Prophet Matching Meets Probing with Commitment}

\author{
Allan Borodin
\thanks{Department of Computer Science, University of Toronto, Toronto, ON, Canada
\texttt{bor@cs.toronto.edu}}
\and
Calum MacRury
\thanks{Department of Computer Science, University of Toronto, Toronto, ON, Canada
\texttt{cmacrury@cs.toronto.edu}}
\and
Akash Rakheja
\thanks{Department of Computer Science, University of Toronto, Toronto, ON, Canada
\texttt{rakhejaakash@gmail.com}}
}

\date{}
\maketitle
% Abstract. Note that this must come before \maketitle.

\begin{abstract}
 
We consider the online stochastic matching problem for bipartite graphs where edges adjacent to an online node must be probed to
determine if they exist, based on known edge probabilities. Our algorithms respect commitment,
in that if a probed edge exists, it must be used in the matching. We study this matching problem subject to a
downward-closed constraint on each online node's allowable edge probes. Our setting generalizes the commonly studied patience (or time-out) constraint  which limits the number of probes that can be made to an online node's adjacent edges. The generality of our setting leads to some  modelling and computational efficiency issues that are not encountered in previous works. To resolve these issues, we introduce a new LP that we prove is a relaxation of an optimal offline probing algorithm (the adaptive benchmark) and which overcomes the limitations of previous LP relaxations. We establish new competitive bounds all of which generalize
the classical non-probing setting when edges do not need to be probed (i.e., exist with certainty).
Specifically, we establish the following competitive ratio results for a  general formulation of edge probing constraints, arbitrary edge weights, and arbitrary edge probabilities:

\begin{enumerate}
 
\item A tight  $\frac{1}{2}$ ratio when the stochastic graph is generated from a known stochastic type graph where the $t^{th}$ online node is drawn independently from a known  distribution $\scr{D}_{\pi(t)}$ and $\pi$ is chosen adversarially. We refer to this setting as the known i.d. stochastic matching problem with adversarial arrivals. 
\item A $1-1/e$ ratio when the stochastic graph is generated from a known stochastic type graph where the $t^{th}$ online node is drawn independently from a  known distribution $\scr{D}_{\pi(t)}$ and $\pi$ is a random permutation. We refer to this setting as the known i.d. stochastic matching problem with random order arrivals. 
\end{enumerate}

Our results improve upon the previous best competitive ratio of $0.46$ in the known i.i.d. setting against the standard adaptive benchmark. Moreover, we are the first to study the prophet secretary matching problem in the context of probing, where
we match the best known classical result.

\end{abstract}

\section{Introduction}
\input{new-intro}

\section{Preliminaries and Our Results}\label{sec:prelim}
\input{new-preliminaries}

\section{Relaxing the Adaptive Benchmark via \ref{LP:new}} \label{sec:relaxation_adaptive_benchmark}
\input{LP_relaxation}

\section{Proving Theorems \ref{thm:known_id_adversarial} and \ref{thm:known_id_ROM}} \label{sec:known_id}
\input{known_id}

\section{Efficiency of Our Algorithms} \label{sec:algorithm_efficiency}
\input{efficient-algorithms}

\section{A Tight Adaptivity Gap} \label{sec:adaptivity_gap}
\input{non_adaptive_negative}

\section{Conclusion and open problems} \label{sec:conclusion}
\input{conclusion-full}

\bibliographystyle{plain}
\bibliography{bibliography}

% Appendix
\appendix

\section{LP Relations} \label{sec:LP_relations}
\input{LP_relations}

\section{Section \ref{sec:known_id} Additions} \label{sec:known_id_additions}

\input{known_id_additions}

%\section{Extended Related Works} \label{appendix:extended_related_work}
%\input{extended_related_works}

\end{document}

%% file: new-intro.tex
Stochastic probing problems are part of the larger area of decision making under uncertainty and more specifically, stochastic optimization. Unlike more standard forms of stochastic optimization, it is not just that there is some stochastic uncertainty in the set of inputs, stochastic probing problems involve inputs that cannot be determined without probing (at some cost and/or within some constraint). Applications of stochastic probing occur naturally in many settings, such as in matching problems where compatibility cannot be determined without some trial or investigation (for example, in online dating and kidney exchange applications).
There is by now an extensive literature for stochastic matching problems.

% For space efficiency, we will give an extended overview of related work in Appendix \ref{appendix:extended_related_work}. Research most directly relating to this paper will appear as we proceed. 

Although we are only considering ``one-sided online bipartite matching'', stochastic matching\footnote{Unfortunately, the term ``stochastic matching'' is also  used to refer to more standard optimization where the input (i.e., edges or vertices) are drawn from some known or unknown distributions but no probing is involved.} was first considered in the context of a general graph by Chen et al. \cite{Chen}. In this problem, the algorithm is presented an adversarially generated {\it stochastic graph} $G = (V, E)$ as input, which has a probability $p_e$ associated with each edge $e$ and a patience (or time-out) parameter $\ell_v$ associated with each vertex $v$. An algorithm probes edges in $E$ within the constraint that at most $\ell_v$ edges are probed incident to any particular vertex $v$. The patience parameter can be viewed as a simple budgetary constraint, where each probe has unit cost and the patience constraint is the budget. When an edge $e$ is probed, it is guaranteed to exist with probability exactly $p_e$. If an edge $(u,v)$ is found to exist, then the algorithm must \textit{commit} to the edge -- that is, it must be added to the current matching (if $u$ and $v$ are currently unmatched). The goal is to maximize the expected size of a matching constructed in this way.

In addition to generalizing the results of Chen et al. to edge weights, Bansal et al. \cite{BansalGLMNR12} introduced a known i.i.d.  bipartite version of the problem where nodes on one side of the partition arrive online and edges adjacent to that node are then probed. In their model, each online vertex is drawn independently and identically from a known distribution. That is, the possible ``type'' of each online node (i.e., its adjacent edge probabilities, edge weights and patience) is known in advance  and the input sequence is then determined i.i.d. from this  known distribution, where the instantiation of each  node is presented to the algorithm upon arrival. In the Bansal et al. model, each offline node has unbounded patience. The match for an online node must be made before the next online arrival.  As in the Chen et al. model, if an edge is probed and confirmed to exist, then it must be included in the current matching (if possible).  This problem is referred to as the {\it online stochastic matching problem} and also referred to as the {\it stochastic rewards problem}.

We study the online  bipartite stochastic matching problem in the more general known i.d. setting.
Specifically, each online vertex is drawn from a (potentially) distinct distribution, and these distributions
are independent. When online vertices arrive adversarially, we generalize the prophet inequality matching problem of Alaei et al. \cite{alaei_2012}. When online vertices arrive in random order, we generalize the prophet secretary matching problem of Ehsani et al. \cite{Ehsani2017}. We note that prophet inequalities give rise to (and in some sense are equivalent to) order oblivious posted price mechanisms, as  first studied in Hajiaghayi et al. \cite{HajiaghayiKS07} and further developed for multi-parameter settings  in Chawla et al. \cite{ChawlaHMS10} and recently in Correa et al. \cite{CorreaFPV19}. Furthermore, we generalize the patience constraint to apply to any downward-closed constraint including budget (equivalently, knapsack) constraints.

Amongst other applications, the online bipartite stochastic matching problem notably models online advertising where the probability of an edge can correspond to the probability of a  purchase in online stores or to pay-per-click revenue in online searching. One may also consider a real estate agent (or owner) of several properties 
working with individual clients each day who first look at properties online.
Each buyer has a value for each property but will not purchase until they see the house in person or until the house has been inspected at which time she will commit to buy.
An agent can only probabilistically  estimate the likelihood of the buyer being satisfied.
But the buyer or agent may  have limited patience (or time or budget) to visit properties and will usually want to do so in a reasonably efficient way. In these two examples, the objective is clearly (respectively) to maximize the revenue of actual sales or clicks and (respectively) to maximize the value of the properties sold.

%% file: new-preliminaries.tex
\label{sec:prelim}
%The \textbf{(bipartite) stochastic matching problem} generalizes the classical (bipartite) matching problem as follows. 
%
%
%A solution to the stochastic matching problem is defined as a \textbf{(offline) probing algorithm}. A probing algorithm
%$\scr{A}$ is given $G=(U,V,E)$ as input, along with its edge weights and edge probabilities, as
%well as its probing constraint $\scr{C}$. Crucially, the edge states $(\st(e))_{e \in E}$
%initially remain hidden to $\scr{A}$. In each round $t \ge 1$, $\scr{A}$ may perform
%a \textbf{probing operation} to reveal the state of an edge $e_{t} =(u_t,v_t) \in E$, subject to
%the requirement that if $(e_{1}, \ldots ,e_{t-1})$ were the previously probed edges,
%then $(e_{1}, \ldots ,e_{t-1}, e_{t}) \in \scr{C}$. The probing algorithm $\scr{A}$
%may be \textbf{adaptive}; that is, in addition to $G$ and $\scr{C}$,
%the decision on whether to probe $e_{t}$ may depend on all the previously probed edges $e_{1}, \ldots ,e_{t-1}$,
%and their revealed states, $\st(e_1),\ldots ,\st(e_{t-1})$. The probing algorithm
%returns a subset of its probed edges which turned out to be active, and which form a matching $\scr{M}$ in $G$.
%Its goal is to maximize $\mb{E}[w(\scr{M})]$, where $w(\scr{M}):= \sum_{e \in \scr{M}} w_{e}$.
%
We define a \textbf{(bipartite) stochastic graph} to be a (simple) bipartite graph $G=(U,V,E)$ with edge weights $(w_{e})_{e \in E}$ and edge probabilities $(p_{e})_{e \in E}$. We refer to $U$ as the \textbf{offline} vertices of $G$ and $V$ as its \textbf{online} nodes. Each $e \in E$ of $G=(U,V,E)$ is \textbf{active} independently with probability $p_e$, and $\st(e) \sim \Ber(p_e)$
corresponds to the \textbf{state} of the edge.
Given an arbitrary set $S$, let $S^{(*)}$ denote the set of all tuples (strings) formed from $S$,
whose entries (characters) are all distinct. Note that we use tuple/string notation and terminology interchangeably. Each $v \in V$ has its own \textbf{(online) probing constraint} $\scr{C}_v \subseteq \partial(v)^{(*)}$, where $\partial(v)$ is the set of edges adjacent to $v$. We make the minimal assumption that $\scr{C}_v$ is \textbf{downward-closed}; that is, if $\bm{e} \in \scr{C}_v$, then any substring or permutation of $\bm{e}$ is also in $\scr{C}_v$.
This includes matroid constraints, as well when each $v \in V$ has a \textbf{budget} $B_v \ge 0$, and \textbf{(edge) probing costs} $(c_{e})_{e \in \partial(v)}$, such that $\bm{e}=(e_1,\ldots,e_k) \in \scr{C}_v$ provided $\sum_{i=1}^{k} c_i \le B_v$. Observe that if each $v$ has uniform probing costs, then this corresponds to the previously discussed case of \textbf{one-sided patience values} $(\ell_v)_{v \in V}$. We shall assume w.l.o.g. that each stochastic graph $G=(U,V,E)$ we work with satisfies $E= U \times V$, as we may always exclude an edge $e$ from being active by setting $p_e = 0$.

A solution to the \textbf{online stochastic matching problem} with \textbf{known i.d. arrivals}
is an \textbf{online probing algorithm}. An online probing algorithm is first presented a stochastic graph $H_{\typ} = (U,B,F)$ with
edge weights $(w_{f})_{f \in F}$, edge probabilities $(p_{f})_{f \in F}$, and online probing constraints $(\scr{C}_{b})_{b \in B}$. We refer to $H_{\typ}$ as a \textbf{(stochastic) type graph}, and the vertices of $B$ as the online  \textbf{type nodes} of $H_{\typ}$. The online probing algorithm does \textit{not} execute on $H_{\typ}$. Instead, 
the adversary fixes a integer $n \ge 1$, and a sequence of distributions $(\scr{D}_{i})_{i=1}^{n}$ supported on $B$,
both of which are known to the online probing algorithm. In the \textbf{adversarial arrival model}, a permutation $\pi$ is generated by an \textbf{oblivious adversary}, in which case $\pi$ is a function of $H_{\typ}$ and $(\scr{D}_i)_{i=1}^{n}$,
and so it \textit{cannot} depend on the instantiation of the online vertices, nor the decisions of the online probing algorithm. 
In the \textbf{random order arrival model}, $\pi$ is generated uniformly at random (u.a.r.), independently of all other randomization.
In either setting, $\pi$ is unknown to the algorithm. For each $t =1, \ldots, n$, vertex $v_{\pi(t)}$ is drawn from $\scr{D}_{\pi(t)}$ and presented to the algorithm, along with its edge weights, probabilities, and online probing constraint. Note that the algorithm is presented
the value $\pi(t)$, and thus learns which distribution $v_{\pi(t)}$ was drawn from. However, the edge states $(\st(e))_{e \in \partial( v_{\pi(t)})}$ initially remain hidden to the algorithm. Instead, using all past available information regarding $v_{\pi(1)}, \ldots ,v_{\pi(t-1)}$, the algorithm must \textbf{probe} the edges of $\partial(v_{\pi(t)})$ to reveal their states, while adhering to the constraint $\scr{C}_{v_{\pi(t)}}$. That is, it may probe a tuple of edges $\bm{e} \in \partial(v_{\pi(t)})^{(*)}$ which satisfies  $\bm{e} \in \scr{C}_{v_{\pi(t)}}$. The algorithm 
operates in the \textbf{probe-commit model}, in which there is a \textbf{commitment} requirement upon probing an edge. Specifically, if an edge $e = (u,v)$ is probed and turns out to be active, then the online probing algorithm must make an irrevocable decision as to whether or not to include $e$ in its matching, prior to probing any subsequent edges. The goal of the algorithm is to build a matching of largest possible expected weight. This definition of commitment
is the one considered by Gupta et al. \cite{GuptaNS16}, and in fact has equivalent power as the previously described Chen et al. \cite{Chen} model.
An online probing algorithm may simply pass on probing an edge if it doesn't intend to add the edge to its matching.

We refer to the randomly generated stochastic graph $G=(U,V,E)$ on which the online
probing algorithm executes as the \textbf{instantiated stochastic graph}, 
and denote $G \sim (H_{\typ}, (\scr{D}_i)_{i=1}^{n})$ to indicate $G$ is drawn from $(H_{\typ}, (\scr{D}_i)_{i=1}^{n})$.
Note that technically $V$ and $E$ are each multi-sets, as multiple copies of a type node $b \in B$ may appear.
In general, it is easy to see we cannot hope to obtain a non-trivial competitive ratio against the
expected weight of an optimal matching of the instantiated stochastic graph\footnote{Consider a single online with a unit patience
value, and $k \ge 1$ offline (unweighted) vertices where each edge $e$ has probability $\frac{1}{k}$ of being present. The expectation of an online probing algorithm will be at most $\frac{1}{k}$ while the expected size of an optimal matching will be $1 - (1-\frac{1}{k})^k \rightarrow 1 -\frac{1}{e}$ as $k \rightarrow \infty$.}.  
The standard approach in the literature is to instead consider the \textbf{offline stochastic matching problem}
and benchmark against an \textit{optimal offline probing algorithm} \cite{BansalGLMNR12, Adamczyk15, BrubachSSX16, BrubachSSX20}.
An \textbf{offline probing algorithm} knows $G=(U,V,E)$, but initially the edge states $(\sta(e))_{e \in E}$ are hidden. It can adaptively probe the edges of $E$ in any order, but must satisfy the probing constraints $(\scr{C}_v)_{v \in V}$ at each step of its execution\footnote{Edges $\bm{e} \in E^{(*)}$ may be probed in order, provided $\bm{e}^{v} \in \scr{C}_v$
for each $v \in V$, where $\bm{e}^{v}$ is the substring of $\bm{e}$ restricted to edges of $\partial(v)$.}
while operating in the same probe-commit model as an online probing algorithm. The goal of an offline probing algorithm is to construct a matching with optimal weight in expectation. We define
the \textbf{adaptive (committal) benchmark} as an optimal offline probing algorithm, and denote
$\OPT(G)$ as the expected weight of its matching when executing on $G$. An alternative weaker benchmark used by Brubach et al. \cite{brubach2021follow,brubach2021conf} is the \textbf{online adaptive benchmark}. This is defined as
an optimal offline probing algorithm which executes on $G$ and whose edge probes respect
some adaptive vertex ordering on $V$. Equivalently, the edge probes
involving each $v \in V$ occur contiguously: if $e'=(u,v') \in E$ is probed after $e = (u,v)$ for $v' \neq v$, then
no edge of $\partial(v)$ is probed following $e'$.
We benchmark against $\mb{E}[\OPT(G)]$, where the expectation is over the randomness in $G \sim (H_{\typ}, (\scr{D}_i)_{i=1}^{n})$. For convenience, we denote $\OPT(H_{\typ}, (\scr{D}_i)_{i=1}^{n}):= \mb{E}[\OPT(G)]$, and refer to $(H_{\typ}, (\scr{D}_i)_{i=1}^{n})$ as a \textbf{known i.d. input}. 

%; that is, upon revealing the state of an edge, an irrevocable
%decision must be made whether to include the edge in its matching before probing any subsequent edges. 

\begin{theorem} \label{thm:known_id_adversarial}
Suppose $(H_{\typ}, (\scr{D}_i)_{i=1}^{n})$ is a known $i.d.$ input,
to which Algorithm \ref{alg:known_id_aom_modified} is given full access.
If $\scr{M}(\pi)$ is the matching returned by the algorithm
when presented the online vertices of $G \sim (H_{\typ}, (\scr{D}_i)_{i=1}^{n})$ in
an adversarial order $\pi:[n] \rightarrow [n]$, then
$\mb{E}[w(\scr{M}(\pi))] \ge \frac{1}{2} \OPT(H_{\typ}, (\scr{D}_i)_{i=1}^{n})$.
\end{theorem}
\begin{remark}
This is a tight bound since the problem generalizes the classic single item prophet inequality for which $\frac{1}{2}$ is an optimal competitive ratio. Recently, Brubach et al. \cite{brubach2021follow,brubach2021conf} independently\footnote{Our work has been motivated by early results of Brubach et al. Since then, we have been improving and extending results independent of improvements and extensions by Brubach et al.}  proved the same competitive ratio against the online adaptive benchmark when $G$ has \textit{unknown} patience values. Our results are incomparable, as their results apply to an unknown patience framework, whereas our results apply to downward-closed online probing constraints, and hold against a stronger benchmark.
\end{remark}

We say that an online probing algorithm is \textbf{non-adaptive},
provided for each $t \in [n]$, the probes of $\partial(v_{\pi(t)})$ are a (randomized) function of $H_{\typ}$,
and the type of arrival $v_{\pi(t)}$. 
In particular, when probing edges adjacent to the online node $v_{\pi(t)}$,  the algorithm does not make use of the 
previously probed edge states of the online nodes $v_{\pi(1)}, \ldots, v_{\pi(t-1)}$, their types, nor the matching decisions made thus far. The algorithm will therefore possibly waste some probes to edges $(u,v_{\pi(t)})$ of which $u \in U$
is unavailable, but it will not violate the matching constraint. 
Our definition generalizes the notion of non-adaptivity introduced 
by Manshadi et al. \cite{ManshadiGS12} in the classic matching problem with known i.i.d. arrivals.

\begin{theorem} \label{thm:known_id_ROM}
Suppose $(H_{\typ}, (\scr{D}_i)_{i=1}^{n})$ is a known $i.d.$ input,
to which Algorithm \ref{alg:known_id_rom_modified} is given full access.
If $\scr{M}$ is the matching returned by the algorithm
when presented the online vertices of $G \sim (H_{\typ}, (\scr{D}_i)_{i=1}^{n})$ in
random order, then $\mb{E}[w(\scr{M})] \ge \left(1 - \frac{1}{e} \right) \OPT(H_{\typ}, (\scr{D}_i)_{i=1}^{n})$.
Moreover, Algorithm \ref{alg:known_id_rom_modified} is non-adaptive.

\end{theorem}
\begin{remark}
The special case of identical distributions with one-sided patience values has been studied in multiple works \cite{BansalGLMNR12, Adamczyk15, BrubachSSX16, BrubachSSX20}, beginning with the $0.12$ competitive ratio of Bansal et al. \cite{BansalGLMNR12},
and for which the previously best known competitive ratio of $0.46$ is due to Brubach et al. \cite{BrubachSSX20}. Note
that all these previous competitive ratios are proven against the adaptive benchmark. Theorem
\ref{thm:known_id_ROM} improves on this result, and also is the first to apply to non-identical distributions, as well as to more
general probing constraints. We also proved this in an earlier arXiv version of this paper \cite{borodin2020}. Recently, Brubach et al. \cite{brubach2021follow} posted an updated version of \cite{Brubach2019} which independently achieves the same competitive ratio for the case of  unknown patience values in the known i.i.d. setting, however again their result is proven
against the online adaptive benchmark. Note that $1-1/e$ is also the best known competitive ratio in the classic matching problem for known i.d. random order arrivals, due to Ehsani et al. \cite{Ehsani2017}.
\end{remark}

In order to discuss the efficiency of our algorithms, we
work in the \textbf{membership oracle model}. An online probing algorithm may make a
\textbf{membership query} to any string $\bm{e} \in \partial(b)^{(*)}$ for $b \in B$,
thus determining in a single operation whether
or not $\bm{e} \in \partial(b)^{(*)}$ is in $\scr{C}_b$. 

%The \textbf{demand oracle model} allows the probing algorithm far more power and is inspired by the classical demand oracle model in the context of the iterative auctions, as originally studied by Blumrosen and Nisan \cite{Blumrosen05,Blumrosen09}. We defer the technical definition of this latter model to Section \ref{sec:algorithm_efficiency}.

\begin{theorem} \label{thm:efficient_known_id}
Suppose that $(H_{\typ}, (\scr{D}_i)_{i=1}^{n})$ is a known $i.d.$ input,
where $H_{\typ}=(U,B,F)$ has downward-closed probing constraints $(\scr{C}_b)_{b \in B}$.
If $|H_{\typ}|$ is the size of $H_{\typ}$ (excluding $(\scr{C}_b)_{b \in B}$),
and $|\scr{D}_{i}|$ is the amount of space needed to encode the distribution $\scr{D}_i$,
then the following claim holds:
\begin{itemize}
    \item In the membership oracle model,  Algorithms \ref{alg:known_id_aom_modified} and \ref{alg:known_id_rom_modified} execute
    in time $\poly(|H_{\typ}|, (|\scr{D}_i|)_{i=1}^{n})$, provided for each $b \in B$, $\scr{C}_b$ is downward-closed.
\end{itemize}
\end{theorem}
\begin{remark}
$|H_{\typ}|$ may be exponentially large in the size of $G \sim (H_{\typ}, (\scr{D}_i)_{i=1}^{n})$,
however for each $\eps >0$, our results can be made to run in time $\poly(|G|, \log(1/\eps))$ 
using Monte Carlo simulation (at a loss of $(1-\eps)$ in performance), assuming we have
oracle access to samples drawn from $(\scr{D}_i)_{i=1}^{n}$.
\end{remark}
An important special case of the online stochastic matching problem is the
case of a \textbf{known stochastic graph}. In this setting, the input $H_{\typ}=(V,B,F)$ satisfies $n= |B|$,
and the distributions $(\scr{D}_i)_{i=1}^{n}$ are all point-mass on \textit{distinct} vertices of $B$. Thus, the online
vertices of $G$ are not randomly drawn, and $G$ is instead 
equal to $H_{\typ}$. The online probing algorithm thus knows the stochastic graph $G$ in advance, but remains unaware of the edge
states $(\st(e))_{e \in E}$, and so it still must sequentially probe the edges to reveal their states. Again, it must operate
in the probe-commit model, and respect the probing constraints $(\scr{C}_v)_{v \in V}$ as well as the arrival order $\pi$ on $V$. Since all the vertices are known a priori, non-adaptivity reduces to requiring that the probes of the algorithm are a function of $G$. We define non-adaptivity for \textit{offline} probing algorithms in the same way.
\begin{corollary} \label{cor:known_stochastic_graph_modified_rom}
Suppose $G=(U,V,E)$ is a known stochastic graph to which Algorithm \ref{alg:known_id_rom_modified}
is given full access. If $\scr{M}$ is the matching returned by the algorithm
when presented the online vertices of $G$ in
random order, then $\mb{E}[w(\scr{M})] \ge \left(1 - \frac{1}{e} \right) \OPT(G)$. Moreover,
Algorithm \ref{alg:known_id_rom_modified} is non-adaptive.
\end{corollary}

\begin{remark}\label{rem:unconstrained}
Gamlath et al. \cite{Gamlath2019} consider an online probing algorithm in the \textbf{unconstrained} 
or \textbf{unbounded patience setting} -- i.e., $\scr{C}_v = \partial(v)^{(*)}$ for $v \in V$ -- when
$G$ is known. Both our algorithm and theirs attain a performance guarantee of $1-1/e$ against different non-standard LPs.
In the special case of unbounded patience,
our LPs take on the same value, as we prove in Proposition \ref{prop:full_patience_equivalence} of Appendix \ref{sec:LP_relations}. Thus, Theorem \ref{thm:known_id_ROM} can be viewed as a generalization of their work to downward-closed online probing constraints
and known i.d. random order arrivals.
\end{remark}

We complement Corollary \ref{cor:known_stochastic_graph_modified_rom}
with a hardness result which applies to \textit{all} non-adaptive probing algorithms (even probing algorithms
which execute offline, and thus do not respect the arrival order $\pi$ of $V$):

\begin{theorem}\label{thm:adaptivity_gap_negative}
In the known stochastic graph setting, no non-adaptive offline probing algorithm can attain an approximation ratio against the adaptive benchmark which is greater than $1-1/e$.
\end{theorem}

\subsection{An Overview of Our Techniques} \label{sec:overview_of_techniques}

%After we attain the appropriate competitive ratios of Theorems \ref{thm:known_id_adversarial} and \ref{thm:known_id_ROM}
%for the known stochastic graph setting, it is easy to modify the online probing algorithms to work for the general case when $G$ is unknown and drawn from a known i.d. instance $(H_{\typ}, (\scr{D}_i)_{i=1}^{n})$. We explain how this is explicitly done in Subsection \ref{subsec:known_id}.

In order to describe the majority of our technical contributions, it suffices to focus on the known stochastic graph setting.
Let us suppose we are presented a stochastic graph $G=(U,V,E)$. For the case of patience values $(\ell_v)_{v \in V}$, a natural solution
is to solve an LP introduced by Bansal et al. \cite{BansalGLMNR12} (see \ref{LP:standard_definition_general} in Appendix \ref{sec:LP_relations}) to obtain fractional values for the edges of $G$, say $(x_{e})_{e \in E}$, such that $x_e$ upper bounds the probability $e$ is probed by the adaptive benchmark. Clearly, $\sum_{e \in \partial(v)} x_{e} \le \ell_{v}$ is a constraint for each $v \in V$, and 
so by applying a dependent rounding algorithm (such as the GKSP algorithm of Gandhi et al. \cite{GandhiGKSP06}), one can round the values $(x_{e})_{e \in \partial(v)}$ to determine $\ell_{v}$ edges of $\partial(v)$ to probe. By probing these edges in a carefully
chosen order, and matching $v$ to the first edge revealed to be active, one can guarantee that each 
$e \in \partial(v)$ is matched with probability reasonably close to $p_e  x_e$. This is the high-level approach used in many previous stochastic matching algorithms (for example \cite{BansalGLMNR12,Adamczyk15, BavejaBCNSX18,BrubachSSX20,brubach2021offline}). However, even for a single online node, \ref{LP:standard_definition_general} overestimates the value of the adaptive benchmark, and so any algorithm designed in
this way will match certain edges with probability strictly less than $p_e  x_e$. This is problematic,
for the value of the match made to $v$ is ultimately compared to $\sum_{e \in \partial(v)} p_{e} w_{e} x_e$, the contribution
of the variables $(x_e)_{e \in \partial(v)}$ to the LP solution. In fact, Brubach et al. \cite{Brubach2019} showed that the ratio between $\OPT(G)$ and an optimum solution to \ref{LP:standard_definition_general} can be as small as $0.544$, so the $1-1/e$ competitive ratio of Theorem \ref{thm:known_id_ROM} cannot be achieved via a comparison to \ref{LP:standard_definition_general},
even for the special case of patience values.

Our approach is to introduce a new configuration LP (\ref{LP:new}) with exponentially many variables which accounts for the many probing strategies available to an arriving vertex $v$ with probing constraint $\scr{C}_v$. For each $\bm{e} =(e_{1}, \ldots , e_{|\bm{e}|}) \in E^{(*)}$, define $g(\bm{e}) := \prod_{i=1}^{|\bm{e}|} (1 - p_{e_i})$,
where $g(\bm{e})$ corresponds to the probability that all the edges of $\bm{e}$ are inactive, and $g(\lambda):=1$
for the empty string/character $\lambda$. We
also define $\bm{e}_{< e_i} := (e_{1}, \ldots ,e_{i-1})$ for each $2 \le i \le |\bm{e}|$,
which we denote by $\bm{e}_{< i}$ when clear, where $\bm{e}_{< 1}:= \lambda$ by convention.
Observe then that $\val(\bm{e}):=\sum_{i=1}^{|\bm{e}|} p_{e_i} w_{e_i} g(\bm{e}_{< i})$
corresponds to the expected weight of the first active edge revealed if $\bm{e}$ 
is probed in order of its indices. For each $v \in V$, we introduce a decision variable $x_{v}(\bm{e})$:
\begin{align}\label{LP:new}
\tag{LP-config}
&\text{maximize} &  \sum_{v \in V} \sum_{\bm{e} \in \scr{C}_v } \val(\bm{e}) \cdot x_{v}(\bm{e}) \\
&\text{subject to} & \sum_{v \in V} \sum_{\substack{ \bm{e} \in \scr{C}_v: \\ (u,v) \in \bm{e}}} 
p_{u,v} \cdot g(\bm{e}_{< (u,v)}) \cdot x_{v}( \bm{e})  \leq 1 && \forall u \in U  \label{eqn:relaxation_efficiency_matching}\\
&& \sum_{\bm{e} \in \scr{C}_v} x_{v}(\bm{e}) = 1 && \forall v \in V,  \label{eqn:relaxation_efficiency_distribution} \\
&&x_{v}( \bm{e}) \ge 0 && \forall v \in V, \bm{e} \in \scr{C}_v
\end{align}
When each $\scr{C}_v$ is downward-closed, \ref{LP:new} can be solved efficiently by using a deterministic separation oracle
for \ref{LP:new_dual}, the dual of \ref{LP:new}, in conjunction with the ellipsoid algorithm \cite{Groetschel,GartnerM},
as we prove in Theorem \ref{thm:LP_solvability} of Section \ref{sec:algorithm_efficiency}. 
Crucially, \ref{LP:new} is also a \textbf{relaxation} of the adaptive benchmark.

\begin{theorem}\label{thm:new_LP_relaxation}
If $G=(U,V,E)$ has downward-closed probing constraints, then $\OPT(G) \le \LPOPT(G)$.
\end{theorem}
\begin{remark}
For the case of patience values, \ref{LP:new} was also recently independently introduced by Brubach et al. \cite{brubach2021follow,brubach2021conf} to design probing algorithms for known i.i.d. arrivals and known i.d. adversarial arrivals. Their competitive ratios are proven against an optimal solution to \ref{LP:new}, which they argue
upper bounds the \textit{online} adaptive benchmark. Theorem \ref{thm:new_LP_relaxation} thus implies that their results in fact hold against the stronger adaptive benchmark.
\end{remark}

In order to prove Theorem \ref{thm:new_LP_relaxation},
the natural approach is to view $x_{v}(\bm{e})$ as the probability that the adaptive benchmark
probes the edges of $\bm{e}$ in order, where $v \in V$ and $\bm{e} \in \scr{C}_v$. 
Let us suppose that hypothetically we could make the following restrictive assumptions
regarding the adaptive benchmark:
\begin{enumerate}[label=(\subscript{P}{{\arabic*}})]
\item If $e=(u,v)$ is probed and $\st(e)=1$, then $e$ is included in the matching, provided $v$ is currently unmatched.
\label{eqn:single_vertex_committal}
\item For each $v \in V$, the edge probes involving $\partial(v)$ are made independently of the edge states $(\sta(e))_{e \in \partial(v)}$. \label{eqn:single_vertex_non_adaptivity}
\end{enumerate}
Observe then that \ref{eqn:single_vertex_committal} and \ref{eqn:single_vertex_non_adaptivity} would imply
that the expected weight of the edge assigned to $v$ is
$\sum_{\bm{e} \in \scr{C}_v } \val(\bm{e}) \cdot x_{v}(\bm{e})$. 
Moreover, the left-hand side of \eqref{eqn:relaxation_efficiency_matching} 
would correspond to the probability $u \in U$ is matched,
so $(x_{v}(\bm{e}))_{v \in V, \bm{e} \in \scr{C}_v}$
would be a feasible solution to \ref{LP:new}, and so
we could upper bound $\OPT(G)$ by $\LPOPT(G)$. Now, if we were working
with the \textit{online} adaptive benchmark, then it is clear that we could assume \ref{eqn:single_vertex_committal} and \ref{eqn:single_vertex_non_adaptivity} simultaneously\footnote{It is clear that we may assume
the adaptive benchmark satisfies \ref{eqn:single_vertex_committal} w.l.o.g., but not \ref{eqn:single_vertex_non_adaptivity}.}
w.l.o.g.  On the other hand, if a probing algorithm does \textit{not} respect an adaptive vertex ordering
on $V$, then the probes involving $v \in V$ will in general depend on $(\sta(e))_{e \in \partial(v)}$. For instance, if $e \in \partial(v)$ is probed and inactive, then perhaps the adaptive benchmark probes $e' =(u,v') \in \partial(v')$ for some $v' \neq v$. If $e'$ is active and thus added to the matching by \ref{eqn:single_vertex_committal}, then the adaptive benchmark can never subsequently probe $(u,v)$ without violating \ref{eqn:single_vertex_committal}, as $u$ is now unavailable to be matched to $v$. Thus, the natural interpretation of the decision variables of \ref{LP:new} does not seem to easily lend itself to a proof
of Theorem \ref{thm:new_LP_relaxation}. 

Our solution is to consider a \textbf{combinatorial relaxation} of the offline stochastic matching problem, which we define to be a new stochastic probing problem on $G$ whose optimal value
$\rOPT(G)$ satisfies $\OPT(G) \le \rOPT(G)$. We refer to this problem as the \textbf{relaxed stochastic matching problem},
a solution to which is a \textbf{relaxed probing algorithm}. Roughly speaking, a relaxed probing algorithm
operates in the same framework as an offline probing algorithm, yet it returns a one-sided matching of the online vertices which matches each offline node at most once \textit{in expectation}. We provide a precise definition in Section \ref{sec:relaxation_adaptive_benchmark}. Crucially, there exists
an \textit{optimal} relaxed probing algorithm which satisfies \ref{eqn:single_vertex_committal} and \ref{eqn:single_vertex_non_adaptivity} simultaneously, which by
the above discussion allows us to conclude that $\rOPT(G) \le \LPOPT(G)$. Since $\OPT(G) \le \rOPT(G)$ by construction,
this implies Theorem \ref{thm:new_LP_relaxation}. Proving the existence of an optimal relaxed probing algorithm with these properties is the most technically challenging part of the paper, and is the content of Lemma \ref{lem:non_adaptive_optimum} of
Section \ref{sec:relaxation_adaptive_benchmark}.

After proving that \ref{LP:new} is a relaxation of the adaptive benchmark, we use it to design online probing algorithms. 
Suppose that we are presented a feasible solution, say $(x_{v}(\bm{e}))_{v \in V, \bm{e} \in \scr{C}_v}$, to \ref{LP:new} for $G$. 
For each $e \in E$, define
\begin{equation}
	\til{x}_{e}:= \sum_{\substack{\bm{e}' \in \scr{C}_v: \\ e \in \bm{e}'}} g(\bm{e}_{ < e}') \cdot x_{v}( \bm{e}'). 
\end{equation}
In order to simplify our notation in the later sections,
we refer to the values $(\til{x}_{e})_{e \in E}$
as the \textbf{(induced) edge variables}
of the solution $(x_{v}(\bm{e}))_{v \in V,\bm{e} \in \scr{C}_v}$.
If we now fix $s \in V$, then we can easily
leverage constraint \eqref{eqn:relaxation_efficiency_distribution} to
design a simple \textit{fixed vertex} probing algorithm which matches
each edge of $e \in \partial(s)$ with probability exactly equal to $p_e \til{x}_e$. Specifically, draw $\bm{e}' \in \scr{C}_s$ with probability $x_{s}(\bm{e}')$. If $\bm{e}' = \lambda$, then return the empty set. Otherwise,
set $\bm{e}' = (e_{1}', \ldots ,e_{k}')$ for $k := |\bm{e}'| \ge 1$, and probe the
edges of $\bm{e}'$ in order. Return the first edge which is revealed to be active, if such an
edge exists. Otherwise, return the empty set. We refer to this algorithm as \textsc{VertexProbe},
and denote its output on the input $(s, \partial(s), (x_{s}(\bm{e}))_{\bm{e} \in \scr{C}_s})$ by
$\textsc{VertexProbe}(s, \partial(s), (x_{s}(\bm{e}))_{ \bm{e} \in \scr{C}_{s}})$.
Observe the following claim, which follows immediately from the definition
of the edge variables, $(\til{x}_{e})_{e \in E}$:
\begin{lemma}\label{lem:fixed_vertex_probe}
Let $G=(U,V,E)$ be a stochastic graph with \ref{LP:new} solution $(x_{v}(\bm{e}))_{v \in V, \partial(v)}$, and whose induced edge variables we denote by $(\til{x}_{e})_{e \in E}$.
If \textsc{VertexProbe} is passed $(s, \partial(s), (x_{s}(\bm{e}))_{ \bm{e} \in \scr{C}_{s}})$, then each $e \in \partial(s)$
is returned by the algorithm with probability $p_e \til{x}_e$. 
\end{lemma}
\begin{definition}\label{def:vertex_probe}
We say that \textsc{VertexProbe} \textbf{commits} to the edge $e=(u,s) \in \partial(s)$, or equivalently the vertex $u \in N(s)$, provided the algorithm outputs $e$ when executing on the fixed node $s \in V$. When it is clear that \textsc{VertexProbe} is being executed on $s$, we say that $s$ commits to $e$ (equivalently the vertex $u$).
\end{definition}
Consider now the simple online probing algorithm, where
$\pi$ is generated either u.a.r. or adversarially.
\begin{algorithm}[H]
\caption{Known Stochastic Graph} 
\label{alg:known_stochastic_graph}
\begin{algorithmic}[1]
\Require a stochastic graph $G=(U,V,E)$.
\Ensure a matching $\scr{M}$ of active edges of $G$.
\State $\scr{M} \leftarrow \emptyset$.
\State Compute an optimal solution of \ref{LP:new} for $G$, say $(x_{v}(\bm{e}))_{v \in V, \bm{e} \in \scr{C}_v}$
\For{$s \in V$ in order based on $\pi$} 
\State Set $e \leftarrow \textsc{VertexProbe}(s, \partial(s), (x_{s}(\bm{e}))_{ \bm{e} \in \scr{C}_{s}})$.
\If{$e=(u,s)$ for some $u \in U$, and $u$ is unmatched} \Comment{this line ensures $e \neq \emptyset$}
\State Add $e$ to $\scr{M}$. \label{line:matched_edge}
\EndIf
\EndFor
\State \Return $\scr{M}$.
\end{algorithmic}
\end{algorithm}
\begin{remark}
Technically, line \eqref{line:matched_edge} 
should occur within the \textsc{VertexProbe} subroutine to adhere
to the probe-commit model, however we express our algorithms in this way for simplicity.
\end{remark}
We observe the following claim, which is easily proven so we omit the argument:
\begin{proposition}\label{prop:known_stochastic_graph}
In the adversarial arrival model, Algorithm \ref{alg:known_stochastic_graph} does not attain a constant competitive ratio. 
In the random order arrival model, Algorithm \ref{alg:known_stochastic_graph} attains a competitive ratio of $1/2$
and the analysis is asymptotically tight.
\end{proposition}
Since the analysis of Algorithm \ref{alg:known_stochastic_graph} cannot be improved in either arrival model, we must modify  the algorithm to prove Theorems \ref{thm:known_id_adversarial} and \ref{thm:known_id_ROM}, even in the known stochastic graph setting.
Our modification involves concurrently applying an appropriate rank one matroid \textbf{contention resolution scheme}
to each offline vertex of $G$, a concept formalized much more generally in the seminal paper by Chekuri, Vondrak, and Zenklusen \cite{Vondrak_2011}. Fix $u \in U$, and observe that constraint \eqref{eqn:relaxation_efficiency_matching} ensures
that $\sum_{e \in \partial(u)} p_e \til{x}_{e} \le 1$. Moreover, if we set
$z_{e}:= p_{e} \til{x}_e$, then observe that as \textsc{VertexProbe} executes on $v$,
each edge $e =(u,v) \in \partial(u)$ is committed to $u$ independently with probability $z_e$. 
On the other hand, there may be many edges which commit to $u$ so we must resolve
which one to take. In Algorithm \ref{alg:known_stochastic_graph}, $u$ is matched greedily to the first online
vertex which commits to it, regardless of how $\pi$ is generated. We apply \textbf{online} and \textbf{random order} contention resolution schemes to ensure that $e$ is matched to $u$ with probability $1/2 \cdot z_e$ when $\pi$ is generated by an adversary, and $(1-1/e) \cdot z_e$ when $\pi$ is generated u.a.r. This allows us to conclude the desired competitive ratios, as $\sum_{e \in E} w_{e} p_{e} \til{x}_e$ upper bounds $\OPT(G)$ by Theorem \ref{thm:new_LP_relaxation}. We review the relevant contention resolution schemes in Section \ref{sec:known_id}, and also extend the argument to the case when $G$ is unknown and instead drawn from the known i.d. input $(H_{\typ}, (\scr{D}_{i})_{i=1}^{n})$, thus proving Theorems \ref{thm:known_id_adversarial} and \ref{thm:known_id_ROM}.

%% file: LP_relaxation.tex
Given a stochastic graph $G=(U,V,E)$, we define the \textbf{relaxed stochastic matching problem}.
A solution to this problem is a \textbf{relaxed probing algorithm} $\scr{A}$,
which operates in the previously described framework of an (offline) probing algorithm. 
That is, $\scr{A}$ is firstly given access to a stochastic graph $G=(U,V,E)$. Initially, the edge states $(\st(e))_{e \in E}$ are unknown to $\scr{A}$, and $\scr{A}$ must adaptivity
probe these edges to reveal their states, while respecting the downward-closed probing constraints $(\scr{C}_v)_{v \in V}$. As in the offline problem, $\scr{A}$ returns a subset $\scr{A}(G)$ of its active probes,
and its goal is to maximize $\mb{E}[w(\scr{A}(G))]$, where $w(\scr{A}(G)):= \sum_{e \in \scr{A}(G)} w_{e}$.
However, unlike before where the output of the probing algorithm was required to be a matching of $G$, we relax
the required properties of $\scr{A}(G)$:
\begin{enumerate}
\item Each $v \in V$ appears in at most one edge of $\scr{A}(G)$.
\item If $N_{u}$ counts the number of edges of $\partial(u)$ which
are included in $\scr{A}(G)$, then $\mb{E}[N_{u}] \le 1$ for each $u \in U$.
\end{enumerate}
We refer to $\scr{A}(G)$ as a \textbf{one-sided matching} of the online nodes.
In constructing $\scr{A}(G)$, $\scr{A}$ must operate in the previously
described probe-commit model. We define the \textbf{relaxed benchmark}
as an optimal relaxed probing algorithm,
and denote its evaluation on $G$
by $\rOPT(G)$. Observe that since any offline probing algorithm is a relaxed probing algorithm,
we have that
\begin{equation} \label{eqn:relaxed_benchmark_upper_bound}
    \OPT(G) \le \rOPT(G).
\end{equation}
We say that $\scr{A}$ is \textbf{non-adaptive},
provided the probes are a (randomized) function of $G$.
Equivalently, $\scr{A}$ is non-adaptive if the probes of $\scr{A}$ are statistically
independent from $(\st(e))_{e \in E}$.
Unlike for the offline stochastic matching problem,
there exists a relaxed probing algorithm which is both optimal \textit{and} non-adaptive:
\begin{lemma} \label{lem:non_adaptive_optimum}
For any stochastic graph $G=(U,V,E)$ with downward-closed probing constraints $(\scr{C}_v)_{v \in V}$,
there exists an optimum relaxed probing algorithm $\scr{B}$ which satisfies the following properties:
\begin{enumerate}[label=(\subscript{Q}{{\arabic*}})]
    \item If $e=(u,v)$ is probed and $\st(e)=1$, then $e$ is included in $\scr{B}(G)$,
    provided $v$ is currently unmatched. \label{eqn:committal}
    \item $\scr{B}$ is non-adaptive on $G$. \label{eqn:non_adaptive}
\end{enumerate}
\end{lemma}
\begin{remark}
Note that \ref{eqn:non_adaptive} implies the hypothetical property \ref{eqn:single_vertex_non_adaptivity}, yet
is much stronger.
\end{remark}
Let us assume Lemma \ref{lem:non_adaptive_optimum} holds for now.
Observe that by considering $\scr{B}$ of Lemma \ref{lem:non_adaptive_optimum},
and defining $x_{v}(\bm{e})$ as the probability that $\scr{B}$
probes the edges of $\bm{e}$ in order for $v\in V$ and $\bm{e} \in \scr{C}_v$, properties \ref{eqn:committal} and \ref{eqn:non_adaptive} ensure that $(x_{v}(\bm{e}))_{v \in V, \bm{e} \in \scr{C}_v}$
is a feasible solution to \ref{LP:new} such that 
\[
	\mb{E}[w(\scr{B}(G))] = \sum_{v \in V} \sum_{\bm{e} \in \scr{C}_v } \val(\bm{e}) \cdot x_{v}(\bm{e}).
\]
Thus, the optimality of $\scr{B}$ implies
that $\rOPT(G) \le \LPOPT(G)$, and so together with \eqref{eqn:relaxed_benchmark_upper_bound},
Theorem \ref{thm:new_LP_relaxation} follows. In fact, \ref{LP:new} is an exact LP formulation of the relaxed stochastic matching problem:
\begin{theorem}\label{thm:LP_relaxation_benchmark_equivalence}
$\rOPT(G) = \LPOPT(G)$.
\end{theorem}
\begin{proof}
Clearly, Theorem \ref{thm:new_LP_relaxation} accounts for one side of the inequality,
so it suffices to show that $\LPOPT(G) \le \rOPT(G)$. Suppose we are presented a feasible solution $(x_{v}(\bm{e}))_{v \in V, \bm{e} \in \scr{C}_v}$ to \ref{LP:new}. Consider then the following algorithm:
\begin{enumerate}
\item $\scr{M} \leftarrow \emptyset$.
\item For each $v \in V$, set $e \leftarrow \textsc{VertexProbe}(v, \partial(v), (x_{v}(\bm{e}))_{\bm{e} \in \scr{C}_v})$.
If $e \neq \emptyset$, then add $e$ to $\scr{M}$.
\item Return $\scr{M}$.
\end{enumerate}
Using Lemma \ref{lem:fixed_vertex_probe}, it is clear that
\[
	\mb{E}[ w(\scr{M})] = \sum_{v \in V} \sum_{\bm{e} \in \scr{C}_v} \val(\bm{e}) \cdot x_{v}(\bm{e}).
\]
Moreover, each vertex $u \in U$ is matched by $\scr{M}$ at most once in expectation, as a consequence of
constraint \eqref{eqn:relaxation_efficiency_matching} of \ref{LP:new}, and so the algorithm satisfies
the required properties of a relaxed probing algorithm.
The proof is therefore complete.
\end{proof}

%Since Theorem \ref{thm:new_LP_relaxation} ensures
%$\rOPT(G) \le \LPOPT(G)$, Theorem \ref{thm:LP_relaxation_benchmark_equivalence} is easily
%proven by designing a relaxed probing algorithm from a solution to \ref{LP:new}
%whose expected value is equal to $\LPOPT(G)$. We defer the details to Appendix
%\ref{sec:LP_relations}.

\subsection{Proving Lemma \ref{lem:non_adaptive_optimum}}

Let us suppose that $G=(U,V,E)$ is a stochastic graph with downward-closed
probing constraints $(\scr{C}_v)_{v \in V}$. In order to prove
Lemma \ref{lem:non_adaptive_optimum}, we must show that there exists
an optimal relaxed probing algorithm which is non-adaptive and satisfies \ref{eqn:committal}.
Our high level approach is to consider an optimal relaxed probing algorithm $\scr{A}$
which satisfies \ref{eqn:committal}, and then to construct a new
non-adaptive algorithm $\scr{B}$ by \textit{stealing} the strategy
of $\scr{A}$, without any loss in performance. More
specifically, we construct $\scr{B}$ by writing down for each $v \in V$
and $\bm{e} \in \scr{C}_v$ the probability that
$\scr{A}$ probes the edges of $\bm{e}$ in order. These
probabilities necessarily satisfy certain inequalities which we make
use of in designing $\scr{B}$. In order to do so, we need a technical randomized rounding
procedure whose precise relevance will become clear in the proof of
Lemma \ref{lem:non_adaptive_optimum}.

Suppose that $\bm{e} \in E^{(*)}$, and recall that $\lambda$ is the empty string/character. For each $j \ge 0$,
denote $\bm{e}_{j}$ as the $j^{th}$ character of $\bm{e}_j$,
where $\bm{e}_j:=\lambda$ when $j=0$ or $j > |\bm{e}|$.
Let us now assume that $(y_{v}(\bm{e}))_{\bm{e} \in \scr{C}_v}$ is a collection of non-negative values
which satisfy $y_{v}(\lambda)=1$, and
\begin{equation} \label{eqn:marginal_distribution}
    \sum_{\substack{e \in \partial(v): \\ (\bm{e}',e) \in \scr{C}_v}} y_{v}(\bm{e}', e) \le y_{v}(\bm{e}'),
\end{equation}
for each $\bm{e}' \in \scr{C}_v$.

\begin{proposition} \label{prop:vertex_round}
Given a collection of values $(y_{v}(\bm{e}))_{\bm{e} \in \scr{C}_v}$ which satisfy
$y_{v}(\lambda)=1$ and \eqref{eqn:marginal_distribution}, there exists a distribution $\scr{D}^v$ supported on $\scr{C}_v$,
such that if $\bm{Y} \sim \scr{D}^v$, then for each $\bm{e} \in \scr{C}_v$ with $k:= |\bm{e}| \ge 1$,
it holds that
\begin{equation}\label{eqn:target_probability}
    \mb{P}[ (\bm{Y}_{1}, \ldots , \bm{Y}_{k}) = (\bm{e}_1, \ldots , \bm{e}_{k})] = y_{v}(\bm{e}),
\end{equation}
where $\bm{Y}_1,\ldots ,\bm{Y}_k$ are the first $k$ characters of $\bm{Y}$.

\end{proposition}

\begin{proof}
First define $\scr{C}^{>}_v:= \{ \bm{e}' \in \scr{C}_v : y_{v}(\bm{e}') > 0\}$, which we observe
is downward-closed since by assumption $\scr{C}_v$ is downward-closed and \eqref{eqn:marginal_distribution} holds. We
prove the proposition for $\scr{C}^{>}_v$, which we then argue implies the proposition holds for $\scr{C}_v$.
Observe now that for each $\bm{e}' \in \scr{C}^{>}_v$, we have that
\begin{equation}\label{eqn:online_vertex_edge_distribution}
\sum_{\substack{e \in \partial(v): \\ (\bm{e}',e) \in \scr{C}^{>}_v}} \frac{y_{v}(\bm{e}', e)}{y_{v}(\bm{e}')} \le 1
\end{equation}
as a result of \eqref{eqn:marginal_distribution} (recall that $y_{v}(\lambda):=1$). We thus define for
each $\bm{e}' \in \scr{C}^{>}_v$,
\begin{equation}\label{eqn:pass_probability}
	z_{v}(\bm{e}'):= 1 - \sum_{\substack{e \in \partial(v): \\ (\bm{e}',e) \in \scr{C}^{>}_v}} \frac{y_{v}(\bm{e}',e)}{y_{v}(\bm{e}')},
\end{equation}
which we observe has the property that $0  \le  z_{v}(\bm{e}') <  1$. The strict inequality follows from the definition
of $\scr{C}^{>}_v$. This leads to the following randomized rounding algorithm,
which we claim outputs a random string $\bm{Y}$ which satisfies the desired properties:
\begin{algorithm}[H]
\caption{VertexRound} \label{alg:general_vertex_rounding}
\begin{algorithmic}[1]
\Require a collection of values $(y_{v}(\bm{e}))_{\bm{e} \in \scr{C}^{>}_v}$ satisfying \eqref{eqn:marginal_distribution} and $y_{v}(\lambda)=1$.
\Ensure a random string $\bm{Y}=(Y_{1},\ldots ,Y_{|\partial(v)|})$ supported on $\scr{C}^{>}_v$.
\State Set $\bm{e}' \leftarrow \lambda$.
\State Initialize $Y_{i}=\lambda$ for each $i=1, \ldots , |\partial(v)|$.
\For{$i=1, \ldots , |\partial(v)|$}
\State Exit the ``for loop'' with probability $z_{v}(\bm{e}')$.
\Comment{pass with a certain probability -- see \eqref{eqn:pass_probability}}
\State Draw $e \in \partial(v)$ satisfying $(\bm{e}',e) \in \scr{C}^{>}_v$ with probability $y_{v}(\bm{e}', e)/ (y_{v}(\bm{e}') \, (1-z_{v}(\bm{e}')))$. \label{line:edge_draw}
\State Set $Y_{i}=e$.
\State $\bm{e}' \leftarrow (\bm{e}',e)$.
\EndFor
\State \Return $\bm{Y}=(Y_{1},\ldots ,Y_{|\partial(v)|})$. 		\Comment{concatenate the edges in order and return the resulting string}
\end{algorithmic}
\end{algorithm}
Clearly, the random string $\bm{Y}$ is supported on $\scr{C}^{>}_v$, thanks to line \ref{line:edge_draw} of Algorithm \ref{alg:general_vertex_rounding}.
We now show that \eqref{eqn:target_probability} holds. As such,
let us first assume $k=1$, and $e \in \partial(v)$ satisfies $(e) \in \scr{C}^{>}_v$. Observe
that
\[
	\mb{P}[Y_{1} = e] = (1 - z_{v}(\lambda)) \frac{ y_{v}(e)}{1 - z_{v}(\lambda)} = y_{v}(e),
\]
as the algorithm does not exit the ``for loop'' with probability $1-z_{v}(\lambda)$, in which case it draws $e$
with probability $y_{v}(e)/(1-z_{v}(\lambda))$. In general, take $k \ge 2$, and assume that for each $\bm{e}' \in \scr{C}^{>}_v$ with $1 \le |\bm{e}'| < k$,
it holds that
\[
	\mb{P}[ (Y_{1},\ldots ,Y_{k}) = \bm{e}'] = y_{v}(\bm{e}').
\]
If we now fix $\bm{e} =(e_{1}, \ldots, e_{k}) \in \scr{C}^{>}_v$ with $|\bm{e}|=k$,
observe that $\bm{e}_{<k}:= (e_{1}, \ldots ,e_{k-1}) \in \scr{C}^{>}_v$, as $\scr{C}^{>}_v$ is downward-closed.
Moreover,
\begin{align*}
	\mb{P}[(Y_{1}, \ldots , Y_{k}) = \bm{e}] &= \mb{P}[ Y_{k} =  e_k \, | \, (Y_{1}, \ldots , Y_{k-1}) = \bm{e}_{<k}] \cdot \mb{P}[ (Y_{1}, \ldots , Y_{k-1}) = \bm{e}_{<k}] \\
					&= \mb{P}[ Y_{k} = e_k \, | \, (Y_{1}, \ldots , Y_{k-1}) = \bm{e}_{<k})] \cdot y_{v}(\bm{e}_{<k}),
\end{align*}
where the last line follows by the induction hypothesis since $\bm{e}_{<k} \in \scr{C}^{>}_v$ is of length $k-1$. 
We know however that
\[
	\mb{P}[ Y_{k} = \bm{e}_k \, | \, (Y_{1}, \ldots , Y_{k-1}) = \bm{e}_{<k}]= (1-z_{v}(\bm{e}_{<k})) \, \frac{y_{v}(\bm{e}_{<k},e_k)}{y_{v}(\bm{e}_{<k}) \, (1-z_{v}(\bm{e}_{<k}))} = \frac{y_{v}(\bm{e}_{<k},e_k)}{y_{v}(\bm{e}_{<k})}.
\]
This is because once we condition on the event $(Y_{1}, \ldots , Y_{k-1}) = \bm{e}_{<k}$, we know that the algorithm
does not exit the ``for loop'' with probability $1 - z_{v}(\bm{e}_{<k})$, in which case it selects $e_{k} \in \partial(v)$ with probability
$y_{v}(\bm{e}_{<k}, e_k)/(y_{v}(\bm{e}_{<k}) \, (1-z_{v}(\bm{e}_{<k})))$, since $(\bm{e}_{<k}, e_k) \in \scr{C}^{>}_v$ by assumption. As such, we have that
\[
	\mb{P}[(Y_{1}, \ldots , Y_{k}) = \bm{e}] = y_{v}(\bm{e}),
\]
and so the proposition holds for $\scr{C}^{>}_v$. 
To complete the argument, observe
that since $\bm{Y}$ is supported on $\scr{C}^{>}_v$, the substrings
of $\bm{Y}$ are also supported on $\scr{C}^{>}_v$, as $\scr{C}^{>}_v$ is downward-closed. Thus, $\bm{Y}$
satisfies \eqref{eqn:target_probability} for the non-empty strings of $\scr{C}_v \setminus \scr{C}^{>}_v$, in addition
to the non-empty strings of $\scr{C}^{>}_v$.
\end{proof}

We are now ready to prove Lemma \ref{lem:non_adaptive_optimum}.

\begin{proof}[Proof of Lemma \ref{lem:non_adaptive_optimum}]

Suppose that $\scr{A}$ is an optimal relaxed probing algorithm which returns
the one-sided matching $\scr{M}$ after executing on the stochastic graph
$G=(U,V,E)$. In a slight abuse of terminology, we say that $e$
is matched by $\scr{A}$, provided $e$ is included in $\scr{M}$.
We shall also make the simplifying assumption that $p_{e} < 1$ for each $e \in E$,
as the proof can be clearly adapted to handle the case when certain edges have
$p_{e}=1$ by restricting which strings of each $\scr{C}_v$ are considered.

Observe that since $\scr{A}$ is optimal, it is clear
that we may assume the following properties hold w.l.o.g. for each $e \in E$:
\begin{enumerate}
    \item $e$ is probed only if $e$ can be added to the currently constructed one-sided
    matching. \label{eqn:probing_only_if}
    \item If $e$ is probed and $\st(e)=1$, then $e$ is included in $\scr{M}$. \label{eqn:if_active_probe}
\end{enumerate}
Thus, in order to prove the lemma, we must find an alternative algorithm $\scr{B}$ which
is non-adaptive, yet continues to be optimal.
To this end, we shall first express $\mb{E}[w(\scr{M}(v))]$ in
a convenient form for each $v \in V$, where  $w(\scr{M}(v))$ is the weight of the edge matched to $v$ (which is $0$ if no match occurs).

Given $v \in V$ and $1 \le i \le |U|$, we define $X_{i}^{v}$ to be the $i^{th}$ edge adjacent to $v$ that is probed by $\scr{A}$.
This is set equal to $\lambda$ by convention, provided no such edge exists. We may
then define $\bm{X}^{v}:=(X^{v}_{1}, \ldots , X^{v}_{|U|})$, and
$\bm{X}_{\le k}^{v} := (X^{v}_1, \ldots , X^{v}_k)$ for each $1 \le k \le |U|$. Moreover, given $\bm{e} =(e_{1}, \ldots ,e_{k}) \in E^{(*)}$ with $k \ge 1$, define $S(\bm{e})$ to be
the event in which $e_{k}$ is the only active edge amongst $e_{1}, \ldots ,e_{k}$.
Observe then that
\[
 \mb{E}[ w(\scr{M}(v))] = \sum_{\substack{\bm{e}=(e_{1}, \ldots ,e_{k}) \in \scr{C}_{v}: \\ k \ge 1}} w_{e_k} \mb{P}[S(\bm{e}) \cap \{\bm{X}^{v}_{\le k} = \bm{e}\}],
 \]
as \eqref{eqn:probing_only_if} and \eqref{eqn:if_active_probe} ensure $v$ is matched to the first probed
edge which is revealed to be active. Moreover,
if $\bm{e}=(e_{1}, \ldots ,e_{k}) \in \scr{C}_v$ for $k \ge 2$, then 
\begin{equation}
     \mb{P}[S(\bm{e}) \cap  \{\bm{X}_{\le k}^{v} = \bm{e} \}] = \mb{P}[\{\st(e_k) =1 \}\cap \{\bm{X}_{\le k}^{v} = \bm{e}\}],
\end{equation}
as \eqref{eqn:probing_only_if} and \eqref{eqn:if_active_probe} ensure $\bm{X}_{\le k}^{v} = \bm{e}$ only if $e_{1}, \ldots , e_{k-1}$ are inactive. Thus, 
\begin{align*}
    \mb{E}[ w(\scr{M}(v))] &= \sum_{\substack{\bm{e}=(e_{1}, \ldots ,e_{k}) \in \scr{C}_{v}: \\ k \ge 1}} w_{e_k} \mb{P}[S(\bm{e}) \cap \{\bm{X}_{\le k}^{v} = \bm{e} \}] \\
                              &= \sum_{\substack{\bm{e}=(e_{1}, \ldots ,e_{k}) \in \scr{C}_{v}: \\ k \ge 1}} w_{e_k} \mb{P}[\{\st(e_{k}) =1\} \cap \{\bm{X}_{\le k}^{v} = \bm{e} \}] \\
                              &= \sum_{\substack{\bm{e}=(e_{1}, \ldots ,e_{k}) \in \scr{C}_{v}: \\ k \ge 1}} w_{e_k} p_{e_k} \mb{P}[\bm{X}_{\le k}^{v} = \bm{e}],
\end{align*}
where the final equality holds since $\scr{A}$ must decide on whether to probe $e_{k}$ prior to revealing $\st(e_k)$.
As a result, after summing over $v \in V$,
\begin{equation}\label{eqn:target_value}
    \mb{E}[ w(\scr{M})] = \sum_{v \in V} \sum_{\substack{\bm{e}=(e_{1}, \ldots ,e_{k}) \in \scr{C}_{v}: \\ k \ge 1}} w_{e_k} p_{e_k} \mb{P}[\bm{X}_{\le k}^{v} = \bm{e} ].
\end{equation}
Our goal is to find a non-adaptive relaxed probing algorithm which matches the value of \eqref{eqn:target_value}.
Thus, for each $v \in V$ and $\bm{e}=(e_{1}, \ldots ,e_{k}) \in \scr{C}_v$ with $k \ge 1$,
define
\[
    x_{v}(\bm{e}):= \mb{P}[\bm{X}_{\le k}^{v} = \bm{e}],
\]
where $x_{v}(\lambda):=1$. Observe now that for each $\bm{e}'=(e_{1}', \ldots , e_{k}') \in \scr{C}_v$,
\begin{equation}\label{eqn:probability_consistency_conditional}
	\sum_{\substack{e \in \partial(v): \\ (\bm{e}',e) \in \scr{C}_v}} \mb{P}[\bm{X}_{\le k+1}^{v} = (\bm{e}',e) \, | \, \bm{X}_{\le k}^{v} = \bm{e}'] \le 1 - p_{e_k'}.
\end{equation}
To see \eqref{eqn:probability_consistency_conditional}, observe that the the left-hand side corresponds to the probability $\scr{A}$
probes some edge $e \in \partial(v)$, given it already probed $\bm{e}'$ in order. On the other hand, if a subsequent edge is probed,
then \eqref{eqn:probing_only_if} and \eqref{eqn:if_active_probe} imply that $e'_k$ must have been inactive, which occurs
independently of the event $\bm{X}_{\le k}^{v} = \bm{e}'$. This explains the right-hand side of \eqref{eqn:probability_consistency_conditional}. Using \eqref{eqn:probability_consistency_conditional}, the values
$(x_{v}(\bm{e}))_{\bm{e} \in \scr{C}_{v}}$ satisfy
\begin{equation}\label{eqn:probability_consistency}
\sum_{\substack{e \in \partial(v): \\ (\bm{e}',e) \in \scr{C}_v}} x_{v}(\bm{e}', e) \le (1 - p_{e'_k}) \cdot x_{v}(\bm{e}'),
\end{equation}
for each $\bm{e}'=(e_{1}', \ldots , e_{k}') \in \scr{C}_v$ with $k \ge 1$. Moreover, clearly
$\sum_{e \in \partial(v)} x_{v}(e) \le 1$.

Given $\bm{e} =(e_1, \ldots ,e_k) \in \scr{C}_v$ for $k \ge 1$,
recall that $\bm{e}_{<k}:=(e_1, \ldots ,e_{k-1})$ where $\bm{e}_{< 1}:= \lambda$ if $k=1$.
Moreover, $g(\bm{e}_{<k}):= \prod_{i=1}^{k-1}(1- p_{e_i})$, where $g(\lambda):=1$.
Using this notation, define for each $\bm{e} \in \scr{C}_v$
\begin{equation}\label{eqn:y_value_definition}
y_{v}(\bm{e}):=
\begin{cases}
     x_{v}(\bm{e})/ g(\bm{e}_{<|\bm{e}|}) & \text{if $|\bm{e}| \ge 1$,} \\
     1 & \text{otherwise.}
\end{cases}
\end{equation}
Observe that \eqref{eqn:probability_consistency} ensures that for each $\bm{e}' \in \scr{C}_v$,
\begin{equation}
    \sum_{\substack{e \in \partial(v): \\ (\bm{e}',e) \in \scr{C}_v}} y_{v}(\bm{e}', e) \le y_{v}(\bm{e}'),
\end{equation}
and $y_{v}(\lambda):=1$.
As a result, Proposition \ref{prop:vertex_round} implies
that for each $v \in V$, there exists a distribution $\scr{D}^{v}$ such that if
$\bm{Y}^{v} \sim \scr{D}^{v}$, then for each $\bm{e} \in \scr{C}_v$ with $|\bm{e}|=k \ge 1$,
\begin{equation}\label{eqn:non_adaptive_target_probability}
    \mb{P}[\bm{Y}^{v}_{\le k} = \bm{e}] = y_{v}(\bm{e}).
\end{equation}
Moreover, $\bm{Y}^{v}$ is drawn independently from the edge states, $(\st(e))_{e \in E}$. Consider now the following algorithm $\scr{B}$, which satisfies the desired properties
\ref{eqn:committal} and \ref{eqn:non_adaptive} of Lemma \ref{lem:non_adaptive_optimum}:
\begin{algorithm}[H]
\caption{Algorithm $\scr{B}$} \label{alg:non_adaptive_relaxed}
\begin{algorithmic}[1]
\Require a stochastic graph $G=(U,V,E)$.
\Ensure a one-sided matching $\scr{N}$ of $G$ of active edges.
\State Set $\scr{N} \leftarrow \emptyset$.
\State Draw $(\bm{Y}^{v})_{v \in V}$ according to the product distribution $\prod_{v \in V} \scr{D}^{v}$.
\For{$v \in V$}
\For{$i=1, \ldots , |\bm{Y}^{v}|$}
\State Set $e \leftarrow \bm{Y}^{v}_{i}$.			\Comment{$\bm{Y}^{v}_{i}$ is the $i^{th}$ edge of $\bm{Y}^{v}$}
\State Probe the edge $e$, revealing $\st(e)$.
\If{$\st(e) =1$ and $v$ is unmatched by $\scr{N}$}
\State Add $e$ to $\scr{N}$.
\EndIf
\EndFor
\EndFor
\State \Return $\scr{N}$.
\end{algorithmic}
\end{algorithm}
Using \eqref{eqn:non_adaptive_target_probability} and the non-adaptivity of $\scr{B}$, it is clear that
for each $v \in V$,
\begin{align*}
    \mb{E}[w(\scr{N}(v))] &= \sum_{\substack{\bm{e}=(e_{1}, \ldots ,e_{k}) \in \scr{C}_{v}: \\ k \ge 1}} w_{e_k} \mb{P}[S(\bm{e})] \cdot \mb{P}[\bm{Y}_{\le k}^{v} = \bm{e}] \\
                          &= \sum_{\substack{\bm{e}=(e_{1}, \ldots ,e_{k}) \in \scr{C}_{v}: \\ k \ge 1}} w_{e_k} p_{e_k} g(\bm{e}_{<k}) y_{v}(\bm{e}) \\
                          &=  \sum_{\substack{\bm{e}=(e_{1}, \ldots ,e_{k}) \in \scr{C}_{v}: \\ k \ge 1}} w_{e_k} p_{e_k} x_{v}(\bm{e}) \\
                          &= \mb{E}[w(\scr{M}(v))].
\end{align*}
Thus, after summing over $v \in V$, it holds that $\mb{E}[w(\scr{N})] = \mb{E}[w(\scr{M})] = \rOPT(G)$,
and so in addition to satisfying \ref{eqn:committal} and \ref{eqn:non_adaptive}, $\scr{B}$ is optimal.
Finally, it is easy to show that each $u \in U$ is matched by $\scr{N}$
at most once in expectation since $\scr{M}$ has this property.
Thus, $\scr{B}$ is a relaxed probing algorithm which is optimal and satisfies the required properties of Lemma \ref{lem:non_adaptive_optimum}.

\end{proof}

%% file: known_id.tex
In this section, we first review rank one contention resolution schemes.
We then prove Theorems \ref{thm:known_id_adversarial} and \ref{thm:known_id_ROM} 
for the case of a known stochastic graph. Finally, in Subsection \ref{subsec:known_id},
we generalize to the case of an arbitrary known i.d. input.

Given $k \ge 1$, consider the ground set $[k]:=\{1, \ldots ,k\}$. Fix $\bm{z} \in [0,1]^{k}$,
and let $R(\bm{z}) \subseteq [k]$ denote the random set
where each $i \in [k]$ is included in $R(\bm{z})$ independently with probability $z_i$.
Let us denote $\scr{P}:=\{ \bm{z} \in [0,1]^{k}: \sum_{i=1}^{k} z_{i} \le 1\}$.
Note that $\scr{P}$ is the convex relaxation of the constraint imposed by
the rank one matroid on $[k]$ (i.e., at most one element of $[k]$ may be selected).

\begin{definition}[Contention Resolution Scheme -- Rank One Matroid -- \cite{Vondrak_2011}]
A \textbf{contention resolution scheme} (CRS) for the rank one matroid on $[k]$ is a (randomized) algorithm $\psi$, which given
$\bm{z} \in \scr{P}$ and $S \subseteq [k]$ as inputs, returns
a single element $\psi_{\bm{z}}(S)$ of $S$. Given $c \in [0,1]$, $\psi$ is said to be $c$-\textbf{selectable},
provided for all $i \in [k]$ and $\bm{z} \in \scr{P}$,
\begin{equation} \label{eqn:selectibility}
    \mb{P}[ i \in \psi_{\bm{z}}(R(\bm{z})) \, | \, i \in R(\bm{z})] \ge c,
\end{equation}
where the probability is over the generation of $R(\bm{z})$,
and the potential randomness used by $\psi$.
\end{definition}
%\begin{remark}
%Observe that if $f:2^{[k]} \rightarrow \mb{R}$ is a monotone linear function,
%then for any $\bm{z} \in \scr{P}$, executing a $c$-selectable CRS $\psi$ yields an element
%$\psi_{\bm{z}}(R(\bm{z})) \in R(\bm{z})$, such that $\mb{E}[ f(\psi_{\bm{z}}(R(\bm{z})))] \ge c  \cdot \mb{E}[ f(R(\bm{z}))]$.
%Thus, $c$-selectable CRS are useful for designing approximation algorithms, in which one
%works with a convex relaxation of the constraint system on $[k]$ (in our case,
%a rank $1$ matroid). Much more general results
%hold, and we refer the reader to the seminal paper by Chekuri, Vondrak, and Zenklusen \cite{Vondrak_2011}.
%\end{remark}
Feldman et al. \cite{Feldman_2016} considered a more restricted class
of contention resolution schemes, called \textbf{online contention resolution schemes} (OCRS).
These are schemes in which $R(\bm{z})$ is \textit{not} known
to the scheme ahead of time. Instead, the elements of $[k]$ are presented to the scheme $\psi$ in adversarial order, where in each step, an arriving $i \in [k]$ reveals if it is in $R(\bm{z})$, at which point $\psi$ must make an irrevocable
decision as to whether it wishes to return $i$ as its output.

In the adversarial arrival model,
we make use of the OCRS recently introduced by Ezra et al. \cite{Ezra_2020},
restricted to the case of a rank one matroid. Note that this scheme is similar to the OCRS considered by Lee and
Singla \cite{Lee2018}, however it has the benefit of not requiring the adversary to 
present the arrival order of $[k]$ to the algorithm upfront.
Given the ground set $[k]=\{1,\ldots k\}$,
suppose  the elements of $[k]$ arrive according to some
permutation $\sigma: [k] \rightarrow [k]$ (i.e., $\sigma(1),\ldots ,\sigma(k)$),
and $\bm{z} \in [0,1]^{k}$ satisfies $\sum_{i=1}^{k} z_i \le 1$. 
Upon the arrival of element $\sigma(t) \in [k]$, compute
\[
	q_{t}:= \frac{1}{2 - \sum_{i=1}^{t-1} z_{\sigma(i)}}.
\]
Observe that $1/2 \le q_{t} \le 1$, as $0 \le \sum_{i=1}^{k} z_i \le 1$, and so the following
OCRS is well-defined:
\begin{algorithm}[H]
\caption{OCRS -- Ezra et al. \cite{Ezra_2020}} 
\label{alg:online_contention_resolution}
\begin{algorithmic}[1]
\Require $\bm{z} \in \scr{P}$, where $\scr{P} \subseteq [0,1]^{k}$. \Comment{$\scr{P}$ is the convex relaxation of the rank one matroid}
\Ensure at most one element of $[k]$.
\For{$t=1,\ldots ,k$}
\If{$\sigma(t) \in R(\bm{z})$}        \Comment{$\sigma(t)$ is in $R(\bm{z})$ with probability $z_{\sigma(t)}$}
\State Compute $q_t$ based on the arrivals $\sigma(1),\ldots ,\sigma(t-1)$.
\State \Return $\sigma(t)$ independently with probability $q_t$. \label{line:OCRS}    
\EndIf
\EndFor
\State \Return $\emptyset$. \Comment{pass on returning an element of $[k]$}
\end{algorithmic}
\end{algorithm}

\begin{theorem}[Ezra et al. \cite{Lee2018}] \label{thm:online_contention_resolution}
Algorithm \ref{alg:online_contention_resolution} is an OCRS for a rank one matroid which
is $1/2$-selectable. 
\end{theorem}
Suppose now we are presented a known stochastic graph $G=(U,V,E)$,
whose online vertices $v_1,\ldots ,v_n$ are presented to the online probing algorithm according
to an adversarially chosen permutation $\pi:[n] \rightarrow [n]$ (i.e., $v_{\pi(1)},\ldots ,v_{\pi(n)}$). Let $(x_{v}(\bm{e}))_{v \in V, \bm{e} \in \scr{C}_v}$ be an optimum solution to \ref{LP:new} for $G$ with induced edge variables $(\til{x}_e)_{e \in E}$.
For each $t \in [n]$ and $u \in U$, define
\begin{equation}
	q_{u,t}:= \frac{1}{2 - \sum_{i=1}^{t-1} z_{u,v_{\pi(i)}}},
\end{equation}
where $z_{e} :=  p_{e} \til{x}_{e}$ for $e \in E$, and $q_{u,1}:=1/2$.
Clearly, $\sum_{v \in V} z_{u, v} \le 1$, by constraint \eqref{eqn:relaxation_efficiency_matching} of \ref{LP:new}, and so $1/2 \le q_{u,t} \le 1$. We consider the following algorithm, is presented $V$ in order $\pi$:
\begin{algorithm}[H]
\caption{Known Stochastic Graph -- AOM -- Modified} 
\label{alg:known_stochastic_graph_aom_modified}
\begin{algorithmic}[1]
\Require a stochastic graph $G=(U,V,E)$.
\Ensure a matching $\scr{M}$ of $G$ of active edges.
\State $\scr{M} \leftarrow \emptyset$.
\State Compute an optimum solution of \ref{LP:new} for $G$, say $(x_{v}(\bm{e}))_{v \in V, \bm{e} \in \scr{C}_v}$.
\For{$t=1, \ldots,n$}
\State Based on the previous arrivals $v_{\pi(1)},\ldots ,v_{\pi(t-1)}$ before $v_{\pi(t)}$, compute values $(q_{u,t})_{u \in U}$.
\State Set $e \leftarrow \textsc{VertexProbe}\left(v_{\pi(t)}, \partial(v_{\pi(t)}), (x_{v_{\pi(t)}}(\bm{e}))_{\bm{e} \in \scr{C}_{v_{\pi(t)}}}\right)$.
\If{$e=(u,v_{\pi(t)})$ for some $u \in U$, and $u$ is unmatched}
\State Add $e$ to $\scr{M}$ independently with probability $q_{u,t}$. \label{line:OCRS_probe}  \Comment{OCRS is used here}
\EndIf
\EndFor
\State \Return $\scr{M}$.
\end{algorithmic}
\end{algorithm}
\begin{proposition} \label{prop:known_stochastic_graph_modified_aom}
Algorithm \ref{alg:known_stochastic_graph_aom_modified} attains a competitive ratio of $1/2$.
\end{proposition}
\begin{proof}
Given $u \in U$, let $\scr{M}(u)$ denote the edge matched to $u$ by $\scr{M}$, where
$\scr{M}(u):=\emptyset$ if no such edge exists. Observe now that if $C(e)$ corresponds to the event in which \textsc{VertexProbe}
commits to $e \in \partial(u)$, then $\mb{P}[C(e)] = p_{e} \til{x}_{e}$ by Lemma \ref{lem:fixed_vertex_probe}.
Moreover, the events $(C(e))_{e \in \partial(u)}$ are independent, and
satisfy
\begin{equation}\label{eqn:within_poly_tope_aom}
	\sum_{e \in \partial(u)} \mb{P}[C(e)] = \sum_{e \in \partial(u)} p_e \til{x}_{e} \le 1,
\end{equation}
by constraint \eqref{eqn:relaxation_efficiency_matching} of \ref{LP:new}.
As such, denote $\bm{z}:=(z_{e})_{e \in \partial(u)}$ where $z_e= p_e\til{x}_e$, and observe
that \eqref{eqn:within_poly_tope_aom} ensures that
$\bm{z} \in \scr{P}$, where $\scr{P}$ is the convex relaxation of the rank one
matroid on $\partial(u)$. Let us denote $R(\bm{z})$ as those those $e \in \partial(u)$ for which $C(e)$ occurs.

If $\psi$ is the OCRS defined in Algorithm \ref{alg:online_contention_resolution}, then we may pass $\bm{z}$ to $\psi$, and process the edges of $\partial(u)$ in the order induced by $\pi$.
Denote the resulting output by $\psi_{\bm{z}}(R(\bm{z}))$.
By coupling the random draws of lines \eqref{line:OCRS} and \eqref{line:OCRS_probe} of Algorithms
\ref{alg:online_contention_resolution} and \ref{alg:known_stochastic_graph_aom_modified}, respectively,
we get that 
\[
	w(\scr{M}(u))= \sum_{e \in \partial(u)} w_{e} \cdot \bm{1}_{[e \in R(\bm{z})]} \cdot \bm{1}_{[e \in \psi_{\bm{z}}(R(\bm{z}))]}
\]
Thus, after taking expectations,
\[
	\mb{E}[w(\scr{M}(u))] = \sum_{e \in \partial(u)} w_{e} \cdot \mb{P}[e \in \psi_{\bm{z}}(R(\bm{z})) \, | \,  e \in R(\bm{z})] \cdot \mb{P}[e \in R(\bm{z})].
\]
Now, Theorem \ref{thm:online_contention_resolution} ensures that for each $e \in \partial(u)$,
$\mb{P}[e \in \psi_{\bm{z}}(R(\bm{z})) \, | \,  e \in R(\bm{z})] \ge 1/2$.
It follows that $\mb{E}[w(\scr{M}(u))] \ge \frac{1}{2} \sum_{e \in \partial(u)} w_{e} p_e \til{x}_e$,
for each $u \in U$. Thus, 
\begin{align*}
	\mb{E}[w(\scr{M})] &=  \sum_{u \in U} \mb{E}[w(\scr{M}(u))] \\
	 &\ge \frac{1}{2} \sum_{e \in E} w_{e} p_e \til{x}_e = \frac{\LPOPT(G)}{2},
\end{align*}
where the equality follows since $(x_{v}(\bm{e}))_{v \in V, \bm{e} \in \scr{C}_v}$ is an optimum solution to \ref{LP:new}.
On the other hand, $\LPOPT(G) \ge \OPT(G)$ by Theorem \ref{thm:new_LP_relaxation},
and so the proof is complete.
\end{proof}
Both Lee and Singla \cite{Lee2018}, as well as Adamczyk and Wlodarczyk \cite{adamczyk2018random},
defined a \textbf{random order contention resolution scheme} (RCRS), which is an OCRS where the elements of $[k]$ arrive in random order. In this definition, the random order is incorporated into the probabilistic computation of selectibility \eqref{eqn:selectibility}. We improve the competitive ratio of Algorithm \ref{alg:known_stochastic_graph} by applying a specific RCRS introduced by Lee and Singla \cite{Lee2018}.
Given the ground set $[k]=\{1,\ldots k\}$,
draw $Y_i \sim [0,1]$ u.a.r. and independently for $i=1, \ldots ,k$. 
\begin{algorithm}[H]
\caption{RCRS -- Lee and Singla \cite{Lee2018}} 
\label{alg:random_contention_resolution}
\begin{algorithmic}[1]
\Require $\bm{z} \in \scr{P}$, where $\scr{P} \subseteq [0,1]^{k}$. 
\Ensure at most one element of $[k]$.
\For{$i \in [k]$ in increasing order of $Y_i$}
\If{$i \in R(\bm{z})$}       
\State \Return $i$ independently with probability $\exp(-Y_{i} \cdot z_{i})$ \label{line:RCRS}
\EndIf
\EndFor
\State \Return $\emptyset$. \Comment{pass on returning an element of $[k]$.}
\end{algorithmic}
\end{algorithm}
\begin{theorem}[Lee and Singla \cite{Lee2018}]\label{thm:random_contention_resolution}
Algorithm \ref{alg:random_contention_resolution} is a $1-1/e$-selectable RCRS
for the case of a rank one matroid.
\end{theorem}
For each $v \in V$, draw $\til{Y}_v \in [0,1]$ independently and uniformly at random. We assume the vertices
of $V$ are presented to the below algorithm in non-decreasing
order, based upon the values $(\til{Y}_v)_{v \in V}$.
\begin{algorithm}[H]
\caption{Known Stochastic Graph -- ROM-- Modified} 
\label{alg:known_stochastic_graph_rom_modified}
\begin{algorithmic}[1]
\Require a stochastic graph $G=(U,V,E)$.
\Ensure a matching $\scr{M}$ of $G$ of active edges.
\State $\scr{M} \leftarrow \emptyset$.
\State Compute an optimum solution of \ref{LP:new} for $G$, say $(x_{v}(\bm{e}))_{v \in V, \bm{e} \in \scr{C}_v}$.
\For{$s \in V$ in increasing order of $\til{Y}_s$} 
\State Set $e \leftarrow \textsc{VertexProbe}(s, \partial(s), (x_{s}(\bm{e}))_{ \bm{e} \in \scr{C}_{s}})$.
\If{$e=(u,s)$ for some $u \in U$, and $u$ is unmatched}
\State Add $e$ to $\scr{M}$ independently with probability $\exp(-\til{Y}_{s} \cdot p_{u,s} \cdot \til{x}_{u,s})$. \label{line:RCRS_probe}
\EndIf
\EndFor
\State \Return $\scr{M}$.
\end{algorithmic}
\end{algorithm}
\begin{proposition}\label{prop:known_stochastic_graph_modified_rom}
Algorithm \ref{alg:known_stochastic_graph_rom_modified} attains a competitive ratio of $1-1/e$.
\end{proposition}
The proof follows almost identically to the proof of Proposition \ref{prop:known_stochastic_graph_modified_aom},
and so we defer it to Appendix \ref{sec:known_id_additions}

\subsection{Extending to Known I.D. Arrivals} \label{subsec:known_id}
Suppose that $(H_{\typ}, (\scr{D}_i)_{i=1}^{n})$ is a known $i.d.$ input,
where $H_{\typ}=(U,B,F)$ has downward-closed online probing constraints $(\scr{C}_b)_{b \in B}$. 
If $G \sim (H_{\typ}, (\scr{D}_i)_{i=1}^{n})$, where $G=(U,V,E)$ has vertices $V=\{v_{1}, \ldots ,v_{n}\}$, then
define $r_{i}(b):=\mb{P}[v_i = b]$ for each $i \in [n]$ and $b \in B$, where we hereby assume
that $r_{i}(b) > 0$. We generalize \ref{LP:new} to account for the distributions $(\scr{D}_i)_{i=1}^{n}$.
For each $i \in [n], b \in B$ and $\bm{e} \in \scr{C}_b$,
we introduce a decision variable $x_{i}( \bm{e} \, || \, b)$
to encode the probability that $v_i$ has type $b$ \textit{and} $\bm{e}$
is the sequence of edges of $\partial(v_i)$ probed by the \textit{relaxed} benchmark.

\begin{align}\label{LP:new_id}
\tag{LP-config-id}
&\text{maximize} &  \sum_{i \in [n], b \in B, \bm{e} \in \scr{C}_b} \val(\bm{e}) \cdot x_{i}(\bm{e} \, ||  \, b)  \\
&\text{subject to} & \sum_{i \in [n], b \in B} \sum_{\substack{ \bm{e} \in \scr{C}_b: \\ (u,b) \in \bm{e}}} 
p_{u,b} \cdot g(\bm{e}_{< (u,b)}) \cdot x_{i}( \bm{e} \, || \, b)  \leq 1 && \forall u \in U  \label{eqn:relaxation_efficiency_matching_id}\\
&& \sum_{\bm{e} \in \scr{C}_b} x_{i}(\bm{e} \, ||  \, b)= r_{i}(b)  && \forall b \in B, i \in [n]  \label{eqn:relaxation_efficiency_distribution_id} \\
&&x_{i}(\bm{e} \, ||  \, b) \ge 0 && \forall b \in B, \bm{e} \in \scr{C}_b, i \in [n]
\end{align}
Let us denote $\LPOPT(H_{\typ}, (\scr{D}_i)_{i=1}^{n})$ as the value of an optimum solution
to \ref{LP:new_id}.
\begin{theorem} \label{thm:known_id_relaxation}
$\OPT(H_{\typ}, (\scr{D}_i)_{i=1}^{n}) \le \LPOPT(H_{\typ}, (\scr{D}_i)_{i=1}^{n})$.
\end{theorem}
One way to prove Theorem \ref{thm:known_id_relaxation} is to use the properties
of the relaxed benchmark on $G$ guaranteed by Lemma \ref{lem:non_adaptive_optimum},
and the above interpretation of the decision variables to argue that
\[
\mb{E}[\rOPT(G)] \le \LPOPT(H_{\typ}, (\scr{D}_i)_{i=1}^{n}),
\]
where $\rOPT(G)$ is the value of the relaxed benchmark on $G$. Specifically, we can interpret \eqref{eqn:relaxation_efficiency_matching_id} as saying that the
relaxed benchmark matches each offline vertex at most once in expectation. Moreover,
\eqref{eqn:relaxation_efficiency_distribution_id} holds by observing that if
$v_i$ is of type $b$, then the relaxed benchmark selects some $\bm{e} \in \scr{C}_{b}$
to probe (note $\bm{e}$ could be the empty-string). We provide a morally equivalent proof of Theorem \ref{thm:known_id_relaxation} in
Appendix \ref{sec:known_id_additions}. Specifically, we consider an optimum solution of \ref{LP:new} with respect to $G$,
and apply a conditioning argument in conjunction with Theorem \ref{thm:new_LP_relaxation}.

Given a feasible solution to \ref{LP:new_id}, 
say $(x_{i}(\bm{e} \, || \, b))_{i \in [n], b \in B,  \bm{e} \in \scr{C}_b}$, for each $u \in U, i \in [n]$ and $b \in B$ define
\begin{equation}\label{eqn:induced_edge_variables_id}
	\til{x}_{u,i}(b):= \sum_{\substack{ \bm{e} \in \scr{C}_b: \\ (u,b) \in \bm{e}}} 
g(\bm{e}_{< (u,b)}) \cdot x_{i}( \bm{e} \, || \, b). 
\end{equation}
We refer to $\til{x}_{u,i}(b)$ as an \textbf{(induced) edge variable}, thus extending the definition from the known stochastic
graph setting. Suppose now that we fix $i \in [n]$ and $b \in B$, and consider the variables, $(x_{i}(\bm{e} \, || \, b))_{\bm{e} \in \scr{C}_b}$. Observe that \eqref{eqn:relaxation_efficiency_distribution_id} ensures that 
\[
    \frac{\sum_{\bm{e} \in \scr{C}_b} x_{i}(\bm{e} \, || \, b)}{ r_{i}(b)} = 1.
\]
Hence, regardless of which type node $v_{i}$ is drawn as,
\[
	\frac{\sum_{\bm{e} \in \scr{C}_{v_i}} x_{i}(\bm{e} \, || \, v_i)}{ r_{i}(v_i)} = 1.
\]
We can therefore generalize \textsc{VertexProbe} as follows. Given vertex $v_i$,
draw $\bm{e}' \in \scr{C}_{v_i}$ with probability $x_{i}(\bm{e}' \, || \, v_i)/r_{i}(v_i)$. If $\bm{e}' = \lambda$, then return the empty-set. Otherwise, set $\bm{e}' = (e_{1}', \ldots ,e_{k}')$ for $k := |\bm{e}'| \ge 1$, and probe the
edges of $\bm{e}'$ in order. Return the first edge which is revealed to be active, if such an
edge exists. Otherwise, return the empty-set. We denote the output of \textsc{VertexProbe} on the input 
$(v_i, \partial(v_i), (x_{i}(\bm{e} \, || \, v_i)/ r_{i}(v_i))_{\bm{e} \in \scr{C}_{v_i}})$ by
$\textsc{VertexProbe}(v_i, \partial(v_i), (x_{i}(\bm{e} \, || \, v_i)/ r_{i}(v_i))_{\bm{e} \in \scr{C}_{v_i}})$.
Define $C(u,v_i)$ as the event in which \textsc{VertexProbe} outputs the edge $(u,v_i)$,
and observe the following extension of Lemma \ref{lem:fixed_vertex_probe}:
\begin{lemma}\label{lem:fixed_vertex_probe_id}
If \textsc{VertexProbe} is passed $\left(v_{i}, \partial(v_i), (x_{i}(\bm{e} \, || \, v_i) / r_{i}(v_i))_{\bm{e} \in \scr{C}_{v_i}}\right)$, then for any $b \in B$ and $u \in U$,
\[
\mb{P}[C(u,v_i) \, | \, v_{i} = b] = \frac{p_{u,b} \cdot \til{x}_{u,i}(b)}{r_{i}(b)}.
\]
\end{lemma}
\begin{remark}
As in Definition \ref{def:vertex_probe}, if $C(u,v_i)$ occurs, then we say that $u$ commits
to $(u,v_i)$ or $v_i$.
\end{remark}
Consider now the generalization of Algorithm \ref{alg:known_stochastic_graph} where
$\pi$ is generated either u.a.r. or adversarially.
\begin{algorithm}[H]
\caption{Known I.D} 
\label{alg:known_id}
\begin{algorithmic}[1]
\Require a known i.d. input $(H_{\typ}, (\scr{D}_i)_{i=1}^{n})$.
\Ensure a matching $\scr{M}$ of active edges of $G \sim (H_{\typ}, (\scr{D}_i)_{i=1}^{n})$.
\State $\scr{M} \leftarrow \emptyset$.
\State Compute an optimum solution of \ref{LP:new_id} for $(H_{\typ}, (\scr{D}_i)_{i=1}^{n})$, say $(x_{i}(\bm{e} \, || \, b))_{i \in [n], b \in B,  \bm{e} \in \scr{C}_b}$.
\For{$t=1, \ldots , n$} 
\State Let $a \in B$ be the type of the current arrival $v_{\pi(t)}$.           \Comment{to simplify notation}
\State Set $e \leftarrow \textsc{VertexProbe}\left(v_{\pi(t)},\partial(v_{\pi(t)}), \left( x_{\pi(t)}(\bm{e} \, || \, a) \cdot r^{-1}_{\pi(t)}(a) \right)_{\bm{e} \in \scr{C}_{a}}\right)$.
\If{$e=(u,v_{\pi(t)})$ for some $u \in U$, and $u$ is unmatched}
\State Add $e$ to $\scr{M}$.
\EndIf
\EndFor
\State \Return $\scr{M}$.

\end{algorithmic}
\end{algorithm}
Similarly, to Algorithm \ref{alg:known_stochastic_graph} of Proposition \ref{prop:known_stochastic_graph},
one can show that Algorithm \ref{alg:known_id} attains a competitive ratio of $1/2$ for random order arrivals. Interestingly, if the distributions $(\scr{D}_{i})_{i=1}^{n}$ are identical -- that is, we work with known i.i.d. arrivals -- then it is relatively  easy to show that
this algorithm's competitive ratio improves to $1-1/e$.
\begin{proposition} \label{prop:known_iid}
If Algorithm \ref{alg:known_id} is presented a known $i.i.d.$ input, say the type graph $H_{\typ}$
together with the distribution $\scr{D}$, then 
$\mb{E}[w(\scr{M})] \ge \left(1 - 1/e \right) \OPT(H_{\typ} , \scr{D})$.
\end{proposition}
\begin{remark}
Proposition \ref{prop:known_iid} is proven explicitly for the case of patience values
in an earlier arXiv version of this paper \cite{borodin2020}. 
\end{remark}
Returning to the case of non-identical distributions, observe
that in the execution of Algorithm \ref{alg:known_id}
the probability that $v_i$ commits to the edge $(u,v_i)$ for $u \in U$ is precisely
\begin{equation}
    z_{u,i} :=  \sum_{b \in B} p_{u,b} \cdot \til{x}_{u,i}(b) = \sum_{b \in B} \sum_{\substack{ \bm{e} \in \scr{C}_b: \\ (u,b) \in \bm{e}}} 
 p_{u,b} \cdot g(\bm{e}_{< (u,b)}) \cdot x_{i}( \bm{e} \, || \, b).
\end{equation}
Moreover, the events $(C(u,v_i))_{i=1}^{n}$ are independent, so this suggests
applying the same contention resolutions schemes as in the known stochastic graph setting. 
We first focus on the adversarial arrival model, where we assume the vertices $v_{1}, \ldots ,v_{n}$
are presented in some unknown order $\pi:[n] \rightarrow [n]$. We make use of the OCRS from before
(Algorithm \ref{alg:online_contention_resolution}).
For each $t \in [n]$ and $u \in U$, define
\begin{equation}
	q_{u,t}:= \frac{1}{2 - \sum_{i=1}^{t-1} z_{u,\pi(i)}},
\end{equation}
where $q_{u,1}:=1/2$. Note that $1/2 \le q_{u,t} \le 1$ as $\sum_{j \in [n]} z_{u, j} \le 1$ by constraint \eqref{eqn:relaxation_efficiency_matching_id} of \ref{LP:new_id}.
We define Algorithm \ref{alg:known_id_aom_modified} by modifying Algorithm \ref{alg:known_id} using the OCRS to ensure that each $i \in [n]$ is matched to $u \in U$ with probability $z_{u,i}/2$. However, to achieve a competitive ratio of $1/2$,
we require the stronger claim that for each type node $a \in B$, the probability $(u,v_i)$ is added to the matching \textit{and} $v_i$ is of type $a$ is lower bounded by $p_{u,a} \til{x}_{u,i}(a)/2$. Crucially, if we condition on $u \in U$ being unmatched when $v_i$ is processed, $v_i$ having type $a$, and $C(u,v_i)$, then the probability the OCRS matches $u$ to $v_i$ does \textit{not} depend on $a$.
As we show below in the proof of Theorem \ref{thm:known_id_adversarial},
this implies the desired lower bound of $p_{u,a} \til{x}_{u,i}(a)/2$, and so
Algorithm \ref{alg:known_id_aom_modified} attains a competitive ratio of $1/2$
by \eqref{eqn:induced_edge_variables_id} and Theorem \ref{thm:known_id_relaxation}.
\begin{algorithm}[H]
\caption{Known I.D. -- AOM -- Modified} 
\label{alg:known_id_aom_modified}
\begin{algorithmic}[1]
\Require a known i.d. input $(H_{\typ}, (\scr{D}_i)_{i=1}^{n})$.
\Ensure a matching $\scr{M}$ of active edges of $G \sim (H_{\typ}, (\scr{D}_t)_{t=1}^{n})$.
\State $\scr{M} \leftarrow \emptyset$.
\State Compute an optimum solution of \ref{LP:new_id} for $(H_{\typ}, (\scr{D}_i)_{i=1}^{n})$, say $(x_{i}(\bm{e} \, || \, b))_{i \in [n], b \in B,  \bm{e} \in \scr{C}_b}$.
\For{$t=1,\ldots ,n$}
\State Let $a \in B$ be the type of the current arrival $v_{\pi(t)}$. 
\State Based on the previous arrivals $v_{\pi(1)},\ldots ,v_{\pi(t-1)}$ before $v_{\pi(t)}$, compute values $(q_{u,t})_{u \in U}$.
\State Set $e \leftarrow \textsc{VertexProbe}\left(v_{\pi(t)}, \partial(v_{\pi(t)}), \left(x_{\pi(t)}(\bm{e} \, || \, a) \cdot r^{-1}_{\pi(t)}(a)  \right)_{\bm{e} \in \scr{C}_{a}} \right)$.
\If{$e=(u,v_t)$ for some $u \in U$, and $u$ is unmatched}
\State Add $e$ to $\scr{M}$ independently with probability $q_{u,t}$.   \label{line:adversarial_contention}
\EndIf
\EndFor
\State \Return $\scr{M}$.
\end{algorithmic}
\end{algorithm}
%However,
%it is not immediately clear this will allow us to attain the desired competitive ratio. Specifically,
%we wish to compare the performance of Algorithm \ref{alg:known_id_aom_modified} to $\LPOPT(H_{\typ}, (\scr{D}_i)_{i=1}^{n})$.
%Thus, we wish to lower bound the probability that $i$ is matched to $u$ \textit{and} has type $b \in B$ by $1/2 \cdot \til{x}_{u,i}(b)$.
%Since the OCRS used (Algorithm \ref{alg:online_contention_resolution} does not depend on the type of $i$, this is implied
%by the guarantee involving $z_{u,i}$, as we show below.

\begin{proof}[Proof of Theorem \ref{thm:known_id_adversarial}]
For notational simplicity, let us assume that $\pi(t)=t$ for each $t \in [n]$,
so that the online vertices arrive in order $v_{1}, \ldots , v_{n}$.
Now, the edge variables $(\til{x}_{u,t}(b))_{u \in U,t \in [n], b \in B}$
satisfy
\[
	\LPOPT(H_{\typ}, (\scr{D}_i)_{i=1}^{n}) = \sum_{u \in U, t \in [n], b \in B} p_{u,b} w_{u,b} \til{x}_{u,t}(b).
\]
Thus, to complete the proof it suffices to show that
\begin{equation} \label{eqn:adversarial_desired_selectibility}
	\mb{P}[ \text{$(u,v_t) \in \scr{M}$ and $v_{t} = b$}] \ge \frac{\til{x}_{u,t}(b)}{2} 
\end{equation}
for each $u \in U, t \in [n]$ and $b \in B$, where we hereby assume w.l.o.g. that $\til{x}_{u,t}(b) > 0$. 
In order to prove this, we first observe that by the same coupling argument used in the proof of Proposition \ref{prop:known_stochastic_graph_modified_aom}, 
\begin{equation} \label{eqn:adversarial_contention_id}
	\mb{P}[(u,v_t) \in \scr{M}] \ge \frac{z_{u,t}}{2} = \frac{1}{2} \sum_{b \in B} p_{u,b} \til{x}_{u,t}(b)
\end{equation}
as a result of the $1/2$-selectability of Algorithm \ref{alg:online_contention_resolution}. 
Let us now define $R_{t}$ as the unmatched vertices of $U$ when $v_t$ arrives. Observe then that
\begin{equation}\label{eqn:conditional_contention_probability}
	\mb{P}[ (u,v_t) \in \scr{M} \, | \, \text{$v_{t} =b, C(u,v_t)$ and $u \in R_t$}] = q_{u,t}.
\end{equation}
Now, $\mb{P}[\text{$v_{t} =b, C(u,v_t)$ and $u \in R_t$}] =  p_{u,b} \cdot \til{x}_{u,t}(b) \cdot \mb{P}[u \in R_t]$,
by Lemma \ref{lem:fixed_vertex_probe_id} and the independence of the events $\{v_{t} =b \} \cap \{ C(u,v_t)\}$ 
and $\{u \in R_t\}$. Thus, by the law of total probability,
\begin{align*}
\sum_{b \in B} p_{u,b} \til{x}_{u,t} q_{u,t} \cdot \mb{P}[u \in R_t] &= \mb{P}[(u,v_t) \in \scr{M}] \\
																				&\ge \frac{z_{u,t}}{2}	\\
																				&=\frac{1}{2} \sum_{b \in B} p_{u,b} \til{x}_{u,t}(b)
\end{align*}
where the second inequality follows from \eqref{eqn:adversarial_contention_id}.
Thus, $q_{u,t} \cdot \mb{P}[u \in R_t] \ge 1/2$, and so combined with \eqref{eqn:conditional_contention_probability},
\eqref{eqn:adversarial_desired_selectibility} follows, thus completing the proof.
\end{proof}

Suppose now that each vertex $v_t$ has an arrival time, say $\til{Y}_{t} \in [0,1]$,
drawn u.a.r. and independently for $t \in [n]$. The values $(\til{Y}_{t})_{t=1}^{n}$
indicate the increasing order in which the vertices $v_{1}, \ldots ,v_{n}$ arrive.

\begin{algorithm}[H]
\caption{Known I.D. -- ROM -- Modified} 
\label{alg:known_id_rom_modified}
\begin{algorithmic}[1]
\Require a known i.d. input $(H_{\typ}, (\scr{D}_t)_{t=1}^{n})$.
\Ensure a matching $\scr{M}$ of active edges of $G \sim (H_{\typ}, (\scr{D}_t)_{t=1}^{n})$.
\State $\scr{M} \leftarrow \emptyset$.
\State Compute an optimum solution of \ref{LP:new_id} for $(H_{\typ}, (\scr{D}_t)_{t=1}^{n})$, say $(x_{t}(\bm{e} \, || \, b))_{t \in [n], b \in B,  \bm{e} \in \scr{C}_b}$.
\For{$t \in [n]$ in increasing order of $\til{Y}_t$} 
\State Set $e \leftarrow \textsc{VertexProbe}\left(v_{t}, \partial(v_t), (x_{t}(\bm{e} \, || \, v_t) / r_{t}(v_t))_{\bm{e} \in \scr{C}_{v_t}} \right)$.
\If{$e=(u,v_t)$ for some $u \in U$, and $u$ is unmatched}
\State Add $e$ to $\scr{M}$ independently with probability $\exp(-\til{Y}_{t} \cdot z_{u,t})$.
\EndIf
\EndFor
\State \Return $\scr{M}$.
\end{algorithmic}
\end{algorithm}
\begin{proof}[Proof of Theorem \ref{thm:known_id_ROM}]
Clearly, Algorithm \ref{alg:known_id_rom_modified} is non-adaptive. The competitive ratio of $1-1/e$ follows by the same coupling argument as in Proposition \ref{prop:known_stochastic_graph_modified_rom},
together with the same observations used in the proof of Theorem \ref{thm:known_id_adversarial}, and so we omit the argument.
\end{proof}

%We conclude the section by noting that \ref{LP:new_id} can be solved in time $\poly(|H_{\typ}|, (r_{t}(b))_{t \in [n], b \in B})$
%in both the membership and demand oracle models, under the assumptions of Theorem \ref{thm:efficient_known_id}.
%The argument follows as in the proof of Theorem \ref{thm:LP_solvability}, with a slight
%adjustment to handle the values $(r_{t}(b))_{t \in [n], b \in B}$ and so we defer the details. 
%The efficiency of Algorithms \ref{alg:known_id_rom_modified} and \ref{alg:known_id_aom_modified} thereby follows,
%as claimed in Theorem \ref{thm:efficient_known_id}.

%% file: efficient-algorithms.tex
In this section, we prove Theorem \ref{thm:efficient_known_id}, thus confirming the
efficiency of the online probing algorithms of Theorems \ref{thm:known_id_adversarial}
and \ref{thm:known_id_ROM}. We show how \ref{LP:new} can be solved efficiently,
as the extension to \ref{LP:new_id} follows identically.

\begin{theorem} \label{thm:LP_solvability}
Suppose that $G=(U,V,E)$ is a stochastic graph with downward-closed probing constraints $(\scr{C}_v)_{v \in V}$.
Given access to a membership oracle, \ref{LP:new} is efficiently solvable in the size of $G$ (excluding the constraints
$(\scr{C}_v)_{v \in V}$).

\end{theorem}
We prove Theorem \ref{thm:LP_solvability} by first considering
the dual of \ref{LP:new}. Note, that in the below LP formulation,
if $\bm{e}=(e_{1}, \ldots , e_{k}) \in \scr{C}_v$,
then we set $e_{i}=(u_{i},v)$ for $i=1, \ldots ,k$ for convenience.
\begin{align}\label{LP:new_dual}
\tag{LP-new-dual}
&\text{minimize} &  \sum_{u \in U} \alpha_{u} + \sum_{v \in V} \beta_{v}  \\
&\text{subject to} & \beta_{v} + \sum_{j=1}^{|\bm{e}|} p_{e_j} \cdot g(\bm{e}_{< j}) \cdot \alpha_{u_j} \ge \sum_{j=1}^{|\bm{e}|} p_{e_j} \cdot w_{e_j} \cdot g( \bm{e}_{ < j}) &&
\forall v \in V, \bm{e} \in \scr{C}_v \\ 
&&  \alpha_{u} \ge 0 && \forall u \in U\\
&& \beta_{v} \in \mb{R} && \forall v \in V
\end{align}

Observe that to prove Theorem \ref{thm:LP_solvability},
it suffices to show that \ref{LP:new_dual} has a deterministic polynomial time separation
oracle, as a consequence of how the ellipsoid algorithm \cite{Groetschel,GartnerM} executes (see \cite{Williamson,Vondrak_2011,Adamczyk2017,Lee2018} for more detail). 

Suppose that we are presented a particular selection of dual variables,   
say $(\alpha_{u})_{u \in U}$ and $(\beta_{v})_{v \in V}$, which may or may not
be a feasible solution to \ref{LP:new_dual}. Our
separation oracle must determine efficiently whether these variables
satisfy all the constraints of \ref{LP:new_dual}. In the case
in which the solution is \textit{infeasible}, the oracle must additionally
return a constraint which is violated. It is clear that we can accomplish this for the non-negativity constraints,
so let us fix a particular $v \in V$ in what follows. 
We wish to determine whether there exists some
$\bm{e}=(e_{1}, \ldots ,e_{k}) \in \scr{C}_v$,
such that if $e_{i}=(u_{i},v)$ for $i=1, \ldots ,k$, then
\begin{equation}\label{eqn:existence_of_string}
\sum_{j=1}^{|\bm{e}|}(w_{e_j} - \alpha_{u_j}) \cdot p_{e_j}  \cdot g(\bm{e}_{< j}) > \beta_{v},
\end{equation}
where the left-hand side of \eqref{eqn:existence_of_string} is $0$ if $\bm{e}=\lambda$.
In order to make this determination, it suffices to solve the following maximization problem.
Given any selection of real values, $(\alpha_{u})_{u \in U}$,
\begin{align}
\text{maximize} \quad &\sum_{i=1}^{|\bm{e}|} (w_{e_i} - \alpha_{u_i}) \cdot p_{e_i} \cdot \prod_{j=1}^{i-1} (1- p_{e_j}) \label{eqn:demand_oracle}\\
\text{subject to} \quad & \bm{e} \in \scr{C}_v
\end{align}
Before we show how \eqref{eqn:demand_oracle} can be solved, we provide a buyer/seller interpretation of the optimization
problem. Assuming first that the edges exist with certainty
(i.e. $p_e \in \{0,1\}$ for all $e \in \partial(v)$), let us suppose a seller is trying to allocate the items of $U$
to a number of buyers. We view the vertex $v$ as a \textit{buyer} who wishes to purchase a subset
of items $S \subseteq U$, based on their valuation function $f(S)$. Assume that
$v$ has \textbf{unit demand}, that is $f(S):= \max_{s \in S} p_{s,v} w_{s,v}$. The values $(\alpha_{s})_{s \in U}$ are viewed as prices the buyer must pay\footnote{See Eden et al. \cite{EdenFFK18} for a buyer/seller interpretation of the classical
\textsc{Ranking} algorithm \cite{KarpVV90} for bipartite matching.}, and the demand oracle returns a solution to $\max_{S \subseteq U} (f(S) - \sum_{s \in S} \alpha_{s})$,
thereby maximizing the utility of $v$. Clearly, for the simple case of a unit-demand buyer,
an optimum assignment is the item  $u \in U$ for which $p_{u,v} w_{u,v} - \alpha_{u}$ is
maximized. 

Returning the setting of arbitrary edge probabilities, even the case of a unit-demand buyer is a non-trivial optimization problem in the stochastic probing framework. Observe that we may view the edge probabilities $(p_{e})_{e \in \partial(v)}$ as modelling the setting when there is uncertainty in whether or not the purchase proposals will succeed; that is, $\st(u,v)=1$, provided the seller agrees to sell item $u$ to buyer $v$. In this interpretation, \eqref{eqn:demand_oracle} is the expected utility of the unit-demand buyer $v$ which purchases the first item $u \in U$ such that $\st(u,v)=1$, at which point $v$ gains utility $w_{u,v} - \alpha_u$. In \cite{borodin2021greedy,borodin2021secretary} \footnote{The title of the conference version \cite{borodin2021secretary} differs
from that of the arXiv version \cite{borodin2021greedy}.}, we show how a buyer can solve
\eqref{eqn:demand_oracle} and use it to design a greedy online probing algorithm. We include
the proof here for completeness. 
\begin{proposition}[\cite{borodin2021greedy,borodin2021secretary}]\label{prop:demand_to_membership_reduction}
If $\scr{C}_v$ is downward-closed, then for any selection of values $(\alpha_u)_{u \in U}$,
\eqref{eqn:demand_oracle} can be solved efficiently, assuming access to a membership query
oracle for $\scr{C}_v$.
\end{proposition}
\begin{proof}
Compute $\til{w}_{e}:= w_{e} - \alpha_u$ for each $e=(u,v) \in \partial(v)$,
and define $P:=\{e \in \partial(v) : \til{w}_e \ge 0\}$.
First observe that if $P = \emptyset$, then \eqref{eqn:demand_oracle}
is maximized by the empty-string $\lambda$. Thus, for now on assume that $P \neq \emptyset$. Since $\scr{C}_v$ is downward-closed, 
it suffices to consider those $\bm{e} \in \scr{C}_v$ whose
edges all lie in $P$. As such, for notational convenience, let us hereby assume
that $\partial(v)=P$. 

For any $\bm{e} \in \scr{C}_v$, let $\bm{e}^{r}$ be the rearrangement of $\bm{e}$, based on the non-increasing order
of the weights $(w_{e})_{e \in \bm{e}}$. Since $\scr{C}_v$ is downward-closed,
we know that $\bm{e}^{r}$ is also in $\scr{C}_v$. Moreover, $\val(\bm{e}^{r}) \ge \val(\bm{e})$
(following observations in \cite{Purohit2019,Brubach2019}).
Hence, let us order the edges of $\partial(v)$ as $e_{1}, \ldots ,e_{m}$,
such that $w_{e_1} \ge \ldots \ge  w_{e_m}$, where $m:=|\partial(s)|$.
Observe then that it suffices to maximize \eqref{eqn:demand_oracle} over
those strings within $\scr{C}_v$ which respect this ordering on $\partial(s)$.
Stated differently, let us denote $\scr{I}_{v}$ as the family of subsets of $\partial(v)$
induced by $\scr{C}_v$, and define the set function $f: 2^{\partial(v)} \rightarrow [0, \infty)$,
where $f(B):= \val(\bm{b})$ for $B=\{b_{1}, \ldots ,b_{|B|}\} \subseteq \partial(v)$,
such that $\bm{b}=(b_{1}, \ldots ,b_{|B|})$ and $w_{b_1} \ge \ldots \ge w_{b_{|B|}}$.
Our goal is then to efficiently maximize $f$ over the set-system $(\partial(v),\scr{I}_v)$.
Observe that $\scr{I}_v$ is downward-closed and that we can simulate oracle access to
$\scr{I}_v$, based on our oracle access to $\scr{C}_v$.

For each $i=0, \ldots ,m-1$, denote $\partial(v)^{>i}:=\{e_{i+1}, \ldots ,e_{m}\}$,
and $\partial(v)^{>m}:= \emptyset$. Moreover, define the family of subsets $\scr{I}_{v}^{>i}:= \{B \subseteq \partial(v)^{>i} : B \cup \{e_i\} \in \scr{I}_v\}$ for each $1 \le i \le m$, and $\scr{I}_{v}^{>0}:= \scr{I}_v$. Observe then that
$(\partial(v)^{>i}, \scr{I}_{v}^{>i})$ is a downward-closed set system, as $\scr{I}_v$ is downward-closed.
Moreover, we may simulate oracle access to $\scr{I}^{>i}_{v}$ based on our oracle access to $\scr{I}_v$.

Denote $\OPT(\scr{I}_{v}^{>i})$ as the maximum value of $f$ over constraints $\scr{I}_{v}^{>i}$.
Observe then that for each $0 \le i \le m-1$, the following recursion holds:
\begin{equation} \label{eqn:dynamical_program}
	\OPT(\scr{I}^{>i}_{v}) :=  
			\max_{j \in \{i+1,\ldots,m\}} ( p_{e_j} \cdot w_{e_j} + (1 - p_{e_j}) \cdot \OPT(\scr{I}_{v}^{>j}) )
\end{equation}
Hence, given access to the values $\OPT(\scr{I}_{v}^{>i+1}), \ldots , \OPT(\scr{I}_{v}^{>m})$,
we can compute $\OPT(\scr{I}^{>i}_v)$ efficiently. Moreover, $\OPT(\scr{I}_{v}^{>m})=0$ by definition. 
Thus, it is clear that we can use \eqref{eqn:dynamical_program} to recover an optimal solution to $f$,
and so the proof is complete.
\end{proof}
We conclude the section by noting that if we are instead
given a known i.d. input $(H_{\typ}, (\scr{D}_i)_{i=1}^{n})$,
then \ref{LP:new_id} can be solved in time $\poly(|H_{\typ}|, (|\scr{D}_i|)_{i=1}^{n})$
using the same strategy, as the same maximization problem \eqref{eqn:demand_oracle} is needed
to separate the dual of \ref{LP:new_id}.
The efficiency of Algorithms \ref{alg:known_id_aom_modified} and \ref{alg:known_id_rom_modified} thereby follows,
as claimed in Theorem \ref{thm:efficient_known_id}.

%% file: non_adaptive_negative.tex
\label{sec:non-adaptive_negative}
Similar to the definition of the adaptive benchmark,
we define the \textbf{non-adaptive benchmark}
as the optimum performance of a non-adaptive probing algorithm
on $G$. That is, $\OPT_{\text{n-adap}}(G):= \sup_{\scr{B}} \mb{E}[ w(\scr{B}(G))]$,
where the supremum is over all offline non-adaptive probing algorithms.
The upper bound (negative result) of Theorem \ref{thm:adaptivity_gap_negative} can
thus be viewed a statement regarding the power of adaptivity. More precisely,
we define the \textbf{adaptivity gap} of the \textbf{stochastic matching problem with
one-sided probing constraints},
as the ratio
\begin{equation}
	\inf_{G} \frac{\OPT_{\text{n-adap}}(G)}{\OPT(G)},
\end{equation}
where the infimum is over all (bipartite) stochastic graphs $G=(U,V,E)$
with \textbf{substring-closed} probing constraints $(\scr{C}_v)_{v \in V}$. Here
$\scr{C}_v$ is closed under substrings if any substring of $\bm{e} \in \scr{C}_v$
is also in $\scr{C}_v$. This is a less restrictive definition than imposing $\scr{C}_v$
must be downward-closed, and is the minimal assumption one needs to ensure stochastic
matching with commitment is well-defined.

We can therefore restate Theorem \ref{thm:adaptivity_gap_negative} in
the following terminology:

\begin{theorem}\label{thm:adaptivity_gap_negative_restatement}
The adaptivity gap of the stochastic matching problem with one-sided probing constraints 
is no smaller than $1-1/e$.
\end{theorem}

Theorem \ref{thm:adaptivity_gap_negative_restatement} follows by considering
a sequence of stochastic graphs. In particular, given $n \ge 1$, consider functions $p=p(n)$ and $s = s(n)$ which satisfy the following:
\begin{enumerate}
\item $p \ll 1/ \sqrt{n}$ and $s \rightarrow \infty$ as $n \rightarrow \infty$.
\item $s \le  p n$ and $s = (1 - o(1)) pn$.
\end{enumerate}
Consider now an unweighted stochastic graph $G_{n}=(U,V,E)$ with unit patience
values, and which satisfies $|U| = s$ and $|V| = n$. Moreover, assume that $p_{u,v} = p$ for all $u \in U$ and $v \in V$.
Observe that $G_n$ corresponds to the bipartite Erdős–Rényi random graph $\mb{G}(s,n,p)$. 
\begin{lemma}\label{lem:committal_benchmark_hardness}

The adaptive benchmark returns a matching of size asymptotically equal to $s$
when executing on $G_n$; that is, $\OPT(G_n) = (1 + o(1)) s$.

\end{lemma}

We omit the proof of Lemma \ref{lem:committal_benchmark_hardness}, as it is routine analysis
of the Erdős–Rényi random graph $\mb{G}(s,n,p)$. Instead, we focus on proving the following
lemma, which together with Lemma \ref{lem:committal_benchmark_hardness}
implies the upper bound of Theorem \ref{thm:adaptivity_gap_negative_restatement}:

\begin{lemma} \label{lem:non_adaptive_hardness}
The non-adaptive benchmark returns in expectation a matching of size at most
$(1+o(1)) \left(1 - \frac{1}{e} \right) s$ when executing on
$G_n$. That is,
\[
	\nOPT(G) \le (1 + o(1)) \left(1 - \frac{1}{e} \right) s.
\] 
\end{lemma}
\begin{proof}
Let $\scr{A}$ be a non-adaptive probing algorithm, which we may assume is deterministic
without loss of generality. As the probes of $\scr{A}$ are
determined independently of the random variables $(\st(e))_{e \in E}$,
we can define $x_{e} \in \{0,1\}$ for each $e \in E$ to indicate whether or not
$\scr{A}$ probes the edge $e$.

Now, if $\scr{A}(G)$ is the matching returned by $\scr{A}$,
then using the independence of the edge states $(\st(e))_{e \in E}$, we get that
\begin{align}
	\mb{P}[\text{$u$ matched by $\scr{A}(G)$}] &= \mb{P}\left[ \cup_{\substack{v \in V: \\ x_{u,v} =1}} \st(u,v)=1 \right] \\
						   &\ge 1 - \prod_{v \in V} (1 - p x_{u,v})
\end{align}
and so, 
\[
	\mb{E}[|\scr{A}(G))|] \le s - \sum_{u \in U} \prod_{v \in V} (1 - p x_{u,v}).
\]
As such, if we can show that 
\[
	\sum_{u \in U} \prod_{v \in V} (1 - p x_{u,v}) \ge (1 - o(1))\frac{s}{e},
\]
then this will imply that
\[
	\mb{E}[|\scr{A}(G)|] \le (1 + o(1))  \left(1 - \frac{1}{e} \right) s.
\]
To see this, first observe that since $p(n) \rightarrow 0$ as $n \rightarrow \infty$, we know
that
\[
	1 - p x_{u,v} = (1 + o(1)) \exp(-p x_{u,v})
\] 
for each $v \in V$. In fact, since $p x_{u,v} \le p$ for all $v \in V$,
the asymptotics are uniform across $V$. More precisely, there exists $C > 0$,
such that for $n$ sufficiently large,
\[
	1 - p x_{u,v} \ge (1 - C p^2) \exp(-p x_{u,v})
\]
for all $v \in V$. As a result, 
\begin{align*}
	\prod_{v \in V} (1 - p x_{u,v}) &\ge (1 - Cp^{2})^{n} \exp\left(-\sum_{v \in V} p x_{u,v}\right)	\\
									&= (1 + o(1))\exp\left(-\sum_{v \in V} p x_{u,v}\right),
\end{align*}
where the second line follows since $p \ll 1/\sqrt{n}$ by assumption. On the other hand, Jensen's inequality
ensures that
\[
	\sum_{u \in U}  \frac{\exp\left(-\sum_{v \in V} p x_{u,v}\right)}{s} \ge \exp\left(- \frac{\sum_{u \in U, v \in V} p x_{u,v}}{n}\right).
\]
However, $\sum_{u \in U} x_{u,v} \le 1$ for each $v \in V$.
Thus, $\sum_{u \in U, v \in V} p x_{u,v} \le p n$, and so
\[
	 \exp\left(- \frac{\sum_{u \in U, v \in V} p x_{u,v}}{s}\right) \ge \exp\left(- \frac{p n}{s}\right) \ge \frac{1}{e},
\]
where the last line follows since $p n \le s$. It follows that
\[
\sum_{u \in U} \prod_{v \in V} (1 - p x_{u,v}) \ge (1 +o(1)) \frac{s}{e},
\]
and so
\[
	\mb{E}[|\scr{A}(G)|] \le (1 + o(1))  \left(1 - \frac{1}{e} \right) s.
\]
As the asymptotics hold uniformly across each deterministic non-adaptive algorithm $\scr{A}$, 
this completes the proof.

\end{proof}

Note that the competitive ratio of Corollary \ref{cor:known_stochastic_graph_modified_rom}
in fact holds whenever the stochastic graph has substring-closed probing constraints. The stronger
downward closed condition is only needed to ensure the efficiency of Algorithm \ref{alg:known_stochastic_graph_rom_modified}.
Thus, Corollary \ref{cor:known_stochastic_graph_modified_rom} and Theorem \ref{thm:adaptivity_gap_negative_restatement}
exactly characterize the adaptivity gap of the stochastic matching problem with one-sided probing constraints:

\begin{corollary}\label{cor:adaptivity_gap}
The adaptivity gap of the stochastic matching problem with one-sided probing constraints is $1-1/e$.
\end{corollary}

%% file: conclusion-full.tex
We have considered the stochastic bipartite matching problem (with probing constraints) in a few settings. As discussed, our results generalize the prophet inequality and prophet secretary matching problems. Our algorithms are polynomial time assuming a mild assumption on the probing constraints
which, in particular, generalizes the standard patience constraints.

%Our main results concern stochastic graphs generated from i.d. distributions for which we obtain an optimal  $\frac{1}{2}$ competitive ratio for adversarial input sequences and a $1-1/e$ competitive ratio for random order input sequences. While the i.d. setting subsumes the known stochastic graph setting, this latter  problem is of independent interest as this is the probing model studied by Chen et al. \cite{Chen}. Unlike the classical  setting, it is not even clear if the Karande et at. \cite{KarandeMT11} result that ``ROM implies known (and unknown) i.i.d'' holds in the known stochastic graph model.  
 
There are some basic questions that are unresolved. Perhaps the most basic question 
which is also unresolved in the classical setting is 
to bridge the gap between the positive $1-1/e$ competitive ratio and inapproximations in the context of known i.d. random order arrivals. In terms of the single item prophet secretary problem (without probing), Correa et al. \cite {CorreaSZ19} obtain a $0.669$ competitive ratio following Azar et al. \cite{AzarCK18} who were the first to surpass the  $1-1/e$ ``barrier''.     
Correa et al. \cite{CorreaSZ19} also establish a .732 inapproximation for the i.d. setting. Our 
adaptivity gap proves the optimality of the $1-1/e$ competitive ratio for non-adaptive algorithms.  
Can we surpass $1-1/e$ in the probing setting for i.d. input arrivals or for the the special case of i.i.d. input arrivals?  Is there a provable difference between  stochastic bipartite matching (with probing constraints) and the 
classical online settings? Can we obtain the same competitive results against an optimal offline {\it  non-committal}  benchmark which respects the probing constraints  but not the commitment constraint. 

One interesting extension of the probing model is to 
allow non-Bernoulli edge random variables 
to describe edge uncertainty. Even for a single online
vertex with full patience, this problem is interesting and has been studied significantly
less (see, \textsc{ProblemMax} in Segev and Singla \cite{Segev2020}). A general understanding of edge uncertainty suggests a possible 
relation between stochastic probing and online algorithms with ML (untrusted) advice (see, for example, Lavastida et al. \cite{LavastidaMRX20}).

%% file: LP_relations.tex
Suppose that we are given an arbitrary stochastic
graph $G=(U,V,E)$. In this section, we
state \ref{LP:standard_definition_general}, the standard
LP in the stochastic matching literature, as introduced by Bansal et al. \cite{BansalGLMNR12},
as well as \ref{LP:full_patience}, the LP introduced by Gamlath et al. \cite{Gamlath2019}.
We then show that \ref{LP:full_patience} and \ref{LP:new} have the same optimum value when $G$ has unbounded
patience.

Consider \ref{LP:standard_definition_general}, which is defined
only when $G$ has patience values $(\ell_v)_{v \in V}$. Here
each $e \in E$ has a variable $x_e$ corresponding to the probability
that the adaptive benchmark probes $e$. 
\begin{align}\label{LP:standard_definition_general}
\tag{LP-std}
&\text{maximize} & \sum_{e \in E} w_{e} \cdot p_{e} \cdot x_{e} \\
&\text{subject to} & \sum_{e \partial(u)} p_{e} \cdot x_{e} & \leq 1 && \forall u \in U \\
& &\sum_{e \in \partial(v)} p_{e} \cdot x_{e} & \leq 1 && \forall v \in V  \\
& &\sum_{e \in \partial(v)} x_{e} & \leq \ell_v && \forall v \in V  \\
& &0 \leq x_{e} &\leq 1 && \forall e \in E.
\end{align}
Gamlath et al. modified \ref{LP:standard_definition_general} for the special case
of unbounded patience by adding in exponentially
many extra constraints. Specifically, for each $v \in V$ and $S \subseteq \partial(v)$, they ensure that
\begin{equation}\label{eqn:svensson_constraints}
	\sum_{e \in S} p_{e} \cdot x_{e} \le 1 - \prod_{e \in S}(1 - p_{e}),
\end{equation}
In the same variable interpretation as \ref{LP:standard_definition_general},
the left-hand side of \eqref{eqn:svensson_constraints} corresponds
to the probability the adaptive benchmark matches an edge of $S \subseteq \partial(v)$, and the right-hand side corresponds to the probability an edge of $S$ is active\footnote{The LP considered by Gamlath et al. in \cite{Gamlath2019}
also places the analogous constraints of \eqref{eqn:svensson_constraints} on the vertices of $U$. That being said, these additional constraints are not used anywhere in the work of Gamlath et al., so we omit them.}. 
\begin{align}\label{LP:full_patience}
\tag{LP-QC}
&\text{maximize} & \sum_{e \in E} w_{e} \cdot p_{e} \cdot x_{e} \\
&\text{subject to} & \sum_{e \in S} p_{e} \cdot x_{e} &\le 1 - \prod_{e \in S}(1 - p_{e}) && \forall v \in V, S \subseteq \partial(v)\\
& &\sum_{e \in \partial(u)} p_{e} \cdot x_{e} & \leq  1&& \forall u \in U  \\ \label{eqn:relxation_matched_full_patience}
& &x_{e} &\ge 0 && \forall e \in E.
\end{align}
Let us denote $\qLPOPT(G)$ as the optimum value of \ref{LP:full_patience}.
\begin{proposition} \label{prop:full_patience_equivalence}
If $G$ has unbounded patience, then $\qLPOPT(G) = \LPOPT(G)$.
\end{proposition}
In order to prove Proposition \ref{prop:full_patience_equivalence}, we make use
of a result of Gamlath et al.  We mention that an almost identical result
is also proven by Costello et al. \cite{costello2012matching} using different techniques.
\begin{theorem}[\cite{Gamlath2019}] \label{thm:costello_svensson_guarantee}
Suppose that $G=(U,V,E)$ is a stochastic graph with unbounded patience, and 
$(x_e)_{e \in E}$ is a solution to \ref{LP:full_patience}. 
For each $v \in V$, there exists an online probing algorithm $\scr{B}_{v}$
whose input is $(v,\partial(v), (x_e)_{e \in \partial(v)})$, and which satisfies
$\mb{P}[\text{$\scr{B}_{v}$ matches $v$ to $e$}] = p_{e} x_e$ for each $e \in \partial(v)$.
\end{theorem}

\begin{proof}[Proof of Proposition \ref{prop:full_patience_equivalence}]

Observe that by Theorem \ref{thm:LP_relaxation_benchmark_equivalence},  in order to prove the claim it suffices
to show that $\qLPOPT(G) = \rOPT(G)$. Clearly, $\rOPT(G) \le \qLPOPT(G)$,
as can be seen by defining $x_{e}$ as the probability that the relaxed benchmark probes
the edge $e \in E$. Thus, we focus on showing that $\qLPOPT(G) \le \rOPT(G)$.

Suppose that $(x_e)_{e \in E}$ is an optimum solution to $\qLPOPT(G)$. 
We design the following algorithm, which we denote by $\scr{B}$:
\begin{enumerate}
\item $\scr{M} \leftarrow \emptyset$.
\item For each $v \in V$, execute  $\scr{B}_{v}$ on $(v,\partial(v), (x_e)_{e \in \partial(v)})$,
where $\scr{B}_v$ is the online probing algorithm of Theorem \ref{thm:costello_svensson_guarantee}.
If $\scr{B}_v$ matches $v$, then let $e'$ be this edge, and add $e'$ to $\scr{M}$
\item Return $\scr{M}$.
\end{enumerate}
Using Theorem \ref{thm:costello_svensson_guarantee}, it is clear that
\[
	\mb{E}[ w(\scr{M})] = \sum_{e \in E} w_e p_e x_e.
\]
Moreover, each vertex $u \in U$ is matched by $\scr{M}$ at most once in expectation, as a consequence of
constraint \eqref{eqn:relxation_matched_full_patience}. As a result, $\scr{B}$ is a relaxed probing algorithm.
Thus, $\qLPOPT(G) =  \sum_{e \in E} w_e p_e x_e \le \rOPT(G)$, and so the proof is complete.
\end{proof}

%% file: known_id_additions.tex
\begin{proof}[Proof of Proposition \ref{prop:known_stochastic_graph_modified_rom}]
Given $u \in U$, let $\scr{M}(u)$ denote the edge matched to $u$ by $\scr{M}$, where
$\scr{M}(u):=\emptyset$ if no such edge exists. Observe now that if $C(e)$ corresponds to the event in which \textsc{VertexProbe}
commits to $e \in \partial(u)$, then $\mb{P}[C(e)] = p_{e} \til{x}_{e}$ by Lemma \ref{lem:fixed_vertex_probe}.
Moreover, the events $(C(e))_{e \in \partial(u)}$ are independent, and
satisfy
\begin{equation}\label{eqn:within_poly_tope}
	\sum_{e \in \partial(u)} \mb{P}[C(e)] = \sum_{e \in \partial(u)} p_e \til{x}_{e} \le 1,
\end{equation}
by constraint \eqref{eqn:relaxation_efficiency_matching} of \ref{LP:new}.
As such, denote $\bm{z}:=(z_{e})_{e \in \partial(u)}$ where $z_e:= p_e\til{x}_e$, and observe
that \eqref{eqn:within_poly_tope} ensures that
$\bm{z} \in \scr{P}$, where $\scr{P}$ is the convex relaxation of the rank $1$
matroid on $\partial(u)$. Let us denote $R(\bm{z})$ as those those $e \in \partial(u)$ for which $C(e)$ occurs.

For each $e = (u,v) \in \partial(u)$, define $Y_{u,v} := \til{Y}_{v}$.
Observe then that the random variables $(Y_{e})_{e \in \partial(u)}$ are independent
and drawn $u.a.r.$ from $[0,1]$. Thus, if $\psi$ is the RCRS defined in Algorithm \ref{alg:random_contention_resolution}, then we may pass $\bm{z}$ to $\psi$, and process the edges of $\partial(u)$ in non-increasing order based on $(Y_{e})_{e \in \partial(u)}$.
Denote the resulting output by $\psi_{\bm{z}}(R(\bm{z}))$.
By coupling the random draws of lines \eqref{line:RCRS} and \eqref{line:RCRS_probe} of Algorithms
\ref{alg:random_contention_resolution} and \ref{alg:known_stochastic_graph_rom_modified}, respectively,
we get that 
\[
	w(\scr{M}(u))= \sum_{e \in \partial(u)} w_{e} \cdot \bm{1}_{[e \in R(\bm{z})]} \cdot \bm{1}_{[e \in \psi_{\bm{z}}(R(\bm{z}))]}
\]
Thus, after taking expectations,
\[
	\mb{E}[w(\scr{M}(u))] = \sum_{e \in \partial(u)} w_{e} \cdot \mb{P}[e \in \psi_{\bm{z}}(R(\bm{z})) \, | \,  e \in R(\bm{z})] \cdot \mb{P}[e \in R(\bm{z})].
\]
Now, Theorem \ref{thm:random_contention_resolution} ensures that for each $e \in \partial(u)$,
$\mb{P}[e \in \psi_{\bm{z}}(R(\bm{z})) \, | \,  e \in R(\bm{z})] \ge \left(1 - \frac{1}{e} \right)$.
It follows that $\mb{E}[w(\scr{M}(u))] \ge \left(1 - \frac{1}{e} \right)\sum_{e \in \partial(u)} w_{e} p_e \til{x}_e$,
for each $u \in U$. Thus, 
\begin{align*}
	\mb{E}[w(\scr{M})] &=  \sum_{u \in U} \mb{E}[w(\scr{M}(u))] \\
	 &\ge \left(1 - \frac{1}{e} \right)\sum_{e \in E} w_{e} p_e \til{x}_e = \left(1 - \frac{1}{e} \right) \LPOPT(G),
\end{align*}
where the equality follows since $(x_{v}(\bm{e}))_{v \in V, \bm{e} \in \scr{C}_v}$ is an optimum solution to \ref{LP:new}.
On the other hand, $\LPOPT(G) \ge \OPT(G)$ by Theorem \ref{thm:new_LP_relaxation},
and so the proof is complete.

\end{proof}

\begin{proof}[Proof of Theorem \ref{thm:known_id_relaxation}]

Suppose that $(H_{\typ}, (\scr{D}_t)_{t=1}^{n})$ is a known i.d. instance,
where $H_{\typ}=(U,B,F)$. Recall that $\scr{C}_b$ corresponds to the online probing
constraint of each type node $b \in B$. For convenience, we denote $\scr{I}:= \sqcup_{b \in B} \scr{C}_b$.
We can then define the following collection of random variables, 
denoted $(X_{t}(\bm{e}))_{t \in [n], \bm{e} \in \scr{I}}$,
based on the following randomized procedure:

\begin{itemize}
\item Draw the instantiated graph $G \sim (H_{\typ}, (\scr{D}_t)_{t=1}^{n})$,
whose vertex arrivals we denote by $v_{1}, \ldots , v_{n}$.
\item Compute an optimum solution of \ref{LP:new} for $G$,
which we denote by $(x_{v_t}(\bm{e}))_{t \in [n], \bm{e} \in \scr{C}_{v_t}}$.
\item For each $t=1, \ldots ,n$ and $\bm{e} \in \scr{I}$, 
set $X_{t}(\bm{e}) = x_{v_t}(\bm{e})$ if $\bm{e} \in \scr{C}_{v_t}$,
otherwise set $X_{t}(\bm{e}) = 0$.
\end{itemize}
Observe then that since by definition $(X_{v_t}(\bm{e}))_{t \in [n], \bm{e} \in \scr{C}_{v_t}}$
is a feasible solution to \ref{LP:new} for $G$, it holds that
for each $t=1, \ldots ,n$
\begin{equation}\label{eqn:online_distribution_iid}
	\sum_{\bm{e} \in \scr{I}} X_{t}(\bm{e}) = 1,
\end{equation}
and
\begin{equation} \label{eqn:offline_matching_iid}
	\sum_{t \in [n], b \in B} \sum_{\substack{ \bm{e} \in \scr{I}: \\ (u,b) \in \bm{e}}} 
p_{u,b} \cdot g(\bm{e}_{< (u,b)}) \cdot X_{t}( \bm{e}) \le 1,
\end{equation}
for each $u \in U$. Moreover, $(X_{t}( \bm{e}))_{t \in [n], \bm{e} \in \scr{C}_{v_t}}$ is
a optimum solution to \ref{LP:new} for $G$, so Theorem \ref{thm:new_LP_relaxation}
implies that
\begin{equation}\label{eqn:known_iid_benchmark_relaxtion}
	\OPT(G) \le \LPOPT(G) = \sum_{t=1}^{n} \sum_{\bm{e} \in \scr{I}} \val(\bm{e}) \cdot X_{t}(\bm{e}). 
\end{equation}
In order to make use of these inequalities in the context of the type graph $H_{\typ}$,
let us first fix a type node $b \in B$ and a string $\bm{e} \in \scr{C}_b$. For each $t \in [n]$, we can then define 
\begin{equation}
	x_{t}(\bm{e} \, || \, b):=\mb{E}[ X_{t}(\bm{e}) \cdot \bm{1}_{[v_{t} = b]}],
\end{equation}
where the randomness is over the generation of $G$. Observe
that by definition of the $(X_{t}(\bm{e}))_{t \in [n], \bm{e} \in \scr{I}}$ values, 
\[
	x_{t}(\bm{e} \, || \, b) =0,
\]
provided $\bm{e} \notin \scr{C}_b$. We claim that $(x_{t}(\bm{e} \, || \, b) )_{b \in B, t \in [n], \bm{e} \in \scr{C}_b}$ is a feasible
solution to \ref{LP:new_id}. To see this, first observe that if we multiply \eqref{eqn:online_distribution_iid} by the indicator random
variable $\bm{1}_{[b_{t} = v]}$, then we get that 
\[
	\sum_{\bm{e} \in \scr{I}} X_{t}(\bm{e}) \cdot \bm{1}_{[v_{t} = b]} = \bm{1}_{[v_{t} = b]}.
\]
As a result, if we take expectations over this equality,
\begin{align*}
	\sum_{\bm{e} \in \scr{I}} x_{t}(\bm{e} \, || \, b) &= \sum_{\bm{e} \in \scr{I}} \mb{E}\left[ X_{t}(\bm{e}) \cdot \bm{1}_{[v_{t} = b]}\right] \\ 
	&= \mb{P}[v_t =b] \\
	&=: r_{t}(b),
\end{align*}
for each $b \in B$ and $t \in [n]$. Let us now fix $u \in U$. Observe that since $X_{t}( \bm{e}) \cdot \bm{1}_{[v_{t} = b]} = X_{t}( \bm{e})$ for
each $\bm{e} \in \scr{C}_b$, \eqref{eqn:offline_matching_iid} ensures that
\begin{equation} \label{eqn:rearrangment}
	\sum_{t \in [n], b \in B} \sum_{\substack{ \bm{e} \in \scr{C}_b: \\ (u,b) \in \bm{e}}} 
p_{u,b} \cdot g(\bm{e}_{< (u,b)}) \cdot X_{t}( \bm{e}) \cdot \bm{1}_{[v_{t} = b]}= \sum_{t \in [n], b \in B} \sum_{\substack{ \bm{e} \in \scr{C}_b: \\ (u,b) \in \bm{e}}} 
p_{u,b} \cdot g(\bm{e}_{< (u,b)}) \cdot X_{t}( \bm{e}) \le 1
\end{equation}
Thus, after taking expectations over \eqref{eqn:rearrangment},
\[
	 \sum_{t \in [n], b \in B} \sum_{\substack{ \bm{e} \in \scr{C}_b: \\ (u,b) \in \bm{e}}} 
p_{u,b} \cdot g(\bm{e}_{< (u,b)}) \cdot x_{t}( \bm{e} \, || \, b)  \le 1,
\]
for each $u \in U$. Since $(x_{t}(\bm{e} \, || \, b))_{t \in [n], b \in B,\bm{e} \in \scr{C}_{b}}$ satisfies these inequalities,
and the variables are clearly all non-negative, 
it follows that $(x_{t}(\bm{e} \, || \, b))_{t \in [n], b \in B,\bm{e} \in \scr{C}_{b}}$ is a feasible solution to \ref{LP:new_id}.
Let us now express the right-hand side of \eqref{eqn:known_iid_benchmark_relaxtion}
as in \eqref{eqn:rearrangment} and take expectations. We then get that
\begin{align*}
	\mb{E}[ \OPT(G)] &\le \sum_{b \in B, t \in [n]} \sum_{\bm{e} \in \scr{I}} \val(\bm{e}) \cdot x_{t}(\bm{e} \, || \, b).
\end{align*}
Now, $\OPT(H_{\typ}, (\scr{D}_i)_{i=1}^{n} = \mb{E}[ \OPT(G)]$ by definition, so since $(x_{t}(\bm{e} \, || \, b))_{b \in B, t \in [n], \bm{e} \in \scr{C}_b}$
is feasible, it holds that
\[
	\OPT(H_{\typ}, (\scr{D}_{i})_{i=1}^{n}) \le \LPOPT_{new-id}(H_{\typ}, (\scr{D}_{i})_{i=1}^{n}),
\]
thus completing the proof.

\end{proof}

%% file: new-main.bbl
\begin{thebibliography}{10}

\bibitem{Adamczyk11}
Marek Adamczyk.
\newblock Improved analysis of the greedy algorithm for stochastic matching.
\newblock {\em Inf. Process. Lett.}, 111(15):731--737, 2011.

\bibitem{Adamczyk2017}
Marek Adamczyk, Fabrizio Grandoni, Stefano Leonardi, and Michal Wlodarczyk.
\newblock When the optimum is also blind: a new perspective on universal
  optimization.
\newblock In {\em ICALP}, 2017.

\bibitem{Adamczyk15}
Marek Adamczyk, Fabrizio Grandoni, and Joydeep Mukherjee.
\newblock Improved approximation algorithms for stochastic matching.
\newblock In Nikhil Bansal and Irene Finocchi, editors, {\em Algorithms - {ESA}
  2015 - 23rd Annual European Symposium, Patras, Greece, September 14-16, 2015,
  Proceedings}, volume 9294 of {\em Lecture Notes in Computer Science}, pages
  1--12. Springer, 2015.

\bibitem{adamczyk2018random}
Marek Adamczyk and Micha{\l} W{\l}odarczyk.
\newblock Random order contention resolution schemes.
\newblock In {\em 2018 IEEE 59th Annual Symposium on Foundations of Computer
  Science (FOCS)}, pages 790--801. IEEE, 2018.

\bibitem{AggarwalGKM11}
Gagan Aggarwal, Gagan Goel, Chinmay Karande, and Aranyak Mehta.
\newblock Online vertex-weighted bipartite matching and single-bid budgeted
  allocations.
\newblock In {\em Proceedings of the Twenty-Second Annual {ACM-SIAM} Symposium
  on Discrete Algorithms, {SODA} 2011, San Francisco, California, USA, January
  23-25, 2011}, pages 1253--1264, 2011.

\bibitem{alaei_2012}
Saeed Alaei, MohammadTaghi Hajiaghayi, and Vahid Liaghat.
\newblock Online prophet-inequality matching with applications to ad
  allocation.
\newblock In {\em Proceedings of the 13th ACM Conference on Electronic
  Commerce}, EC '12, page 18–35, New York, NY, USA, 2012. Association for
  Computing Machinery.

\bibitem{Asadpour2016}
Arash Asadpour and Hamid Nazerzadeh.
\newblock Maximizing stochastic monotone submodular functions.
\newblock {\em Management Science}, 62(8):2374--2391, 2016.

\bibitem{AzarCK18}
Yossi Azar, Ashish Chiplunkar, and Haim Kaplan.
\newblock Prophet secretary: Surpassing the 1-1/e barrier.
\newblock In {\'{E}}va Tardos, Edith Elkind, and Rakesh Vohra, editors, {\em
  Proceedings of the 2018 {ACM} Conference on Economics and Computation,
  Ithaca, NY, USA, June 18-22, 2018}, pages 303--318. {ACM}, 2018.

\bibitem{BansalGLMNR12}
Nikhil Bansal, Anupam Gupta, Jian Li, Juli{\'{a}}n Mestre, Viswanath Nagarajan,
  and Atri Rudra.
\newblock When {LP} is the cure for your matching woes: Improved bounds for
  stochastic matchings.
\newblock {\em Algorithmica}, 63(4):733--762, 2012.

\bibitem{BavejaBCNSX18}
Alok Baveja, Amit Chavan, Andrei Nikiforov, Aravind Srinivasan, and Pan Xu.
\newblock Improved bounds in stochastic matching and optimization.
\newblock {\em Algorithmica}, 80(11):3225--3252, Nov 2018.

\bibitem{Blumrosen05}
Liad Blumrosen and Noam Nisan.
\newblock On the computational power of iterative auctions.
\newblock In John Riedl, Michael~J. Kearns, and Michael~K. Reiter, editors,
  {\em Proceedings 6th {ACM} Conference on Electronic Commerce (EC-2005),
  Vancouver, BC, Canada, June 5-8, 2005}, pages 29--43. {ACM}, 2005.

\bibitem{Blumrosen09}
Liad Blumrosen and Noam Nisan.
\newblock On the computational power of demand queries.
\newblock {\em {SIAM} J. Comput.}, 39(4):1372--1391, 2009.

\bibitem{borodin2020greedy}
Allan Borodin, Calum MacRury, and Akash Rakheja.
\newblock Greedy approaches to online stochastic matching.
\newblock {\em CoRR}, abs/2008.09260, 2020.

\bibitem{Bradac2019}
Domagoj Bradac, Sahil Singla, and Goran Zuzic.
\newblock (near) optimal adaptivity gaps for stochastic multi-value probing.
\newblock In Dimitris Achlioptas and L{\'{a}}szl{\'{o}}~A. V{\'{e}}gh, editors,
  {\em Approximation, Randomization, and Combinatorial Optimization. Algorithms
  and Techniques, {APPROX/RANDOM} 2019, September 20-22, 2019, Massachusetts
  Institute of Technology, Cambridge, MA, {USA}}, volume 145 of {\em LIPIcs},
  pages 49:1--49:21. Schloss Dagstuhl - Leibniz-Zentrum f{\"{u}}r Informatik,
  2019.

\bibitem{Brubach2019}
Brian Brubach, Nathaniel Grammel, and Aravind Srinivasan.
\newblock Vertex-weighted online stochastic matching with patience constraints.
\newblock {\em CoRR}, abs/1907.03963, 2019.

\bibitem{BrubachSSX16}
Brian Brubach, Karthik~Abinav Sankararaman, Aravind Srinivasan, and Pan Xu.
\newblock New algorithms, better bounds, and a novel model for online
  stochastic matching.
\newblock In {\em 24th Annual European Symposium on Algorithms, {ESA} 2016,
  August 22-24, 2016, Aarhus, Denmark}, pages 24:1--24:16, 2016.

\bibitem{BrubachSSX20}
Brian Brubach, Karthik~Abinav Sankararaman, Aravind Srinivasan, and Pan Xu.
\newblock Attenuate locally, win globally: Attenuation-based frameworks for
  online stochastic matching with timeouts.
\newblock {\em Algorithmica}, 82(1):64--87, 2020.

\bibitem{ChawlaGKM20}
Shuchi Chawla, Kira Goldner, Anna~R. Karlin, and J.~Benjamin Miller.
\newblock Non-adaptive matroid prophet inequalities.
\newblock {\em CoRR}, abs/2011.09406, 2020.

\bibitem{ChawlaHMS10}
Shuchi Chawla, Jason~D. Hartline, David~L. Malec, and Balasubramanian Sivan.
\newblock Multi-parameter mechanism design and sequential posted pricing.
\newblock In Leonard~J. Schulman, editor, {\em Proceedings of the 42nd {ACM}
  Symposium on Theory of Computing, {STOC} 2010, Cambridge, Massachusetts, USA,
  5-8 June 2010}, pages 311--320. {ACM}, 2010.

\bibitem{Chen}
Ning Chen, Nicole Immorlica, Anna~R. Karlin, Mohammad Mahdian, and Atri Rudra.
\newblock Approximating matches made in heaven.
\newblock In {\em Proceedings of the 36th International Colloquium on Automata,
  Languages and Programming: Part I}, ICALP '09, pages 266--278, 2009.

\bibitem{CorreaFPV19}
Jos{\'{e}}~R. Correa, Patricio Foncea, Dana Pizarro, and Victor Verdugo.
\newblock From pricing to prophets, and back!
\newblock {\em Oper. Res. Lett.}, 47(1):25--29, 2019.

\bibitem{CorreaSZ19}
Jos{\'{e}}~R. Correa, Raimundo Saona, and Bruno Ziliotto.
\newblock Prophet secretary through blind strategies.
\newblock In {\em Proceedings of the Thirtieth Annual {ACM-SIAM} Symposium on
  Discrete Algorithms, {SODA} 2019, San Diego, California, USA, January 6-9,
  2019}, pages 1946--1961, 2019.

\bibitem{DeanGV05}
Brian~C. Dean, Michel~X. Goemans, and Jan Vondr{\'{a}}k.
\newblock Adaptivity and approximation for stochastic packing problems.
\newblock In {\em Proceedings of the Sixteenth Annual {ACM-SIAM} Symposium on
  Discrete Algorithms, {SODA} 2005, Vancouver, British Columbia, Canada,
  January 23-25, 2005}, pages 395--404, 2005.

\bibitem{DeanGV08}
Brian~C. Dean, Michel~X. Goemans, and Jan Vondr{\'{a}}k.
\newblock Approximating the stochastic knapsack problem: The benefit of
  adaptivity.
\newblock {\em Math. Oper. Res.}, 33(4):945--964, 2008.

\bibitem{DJK2013}
Nikhil~R. Devanur, Kamal Jain, and Robert~D. Kleinberg.
\newblock Randomized primal-dual analysis of ranking for online bipartite
  matching.
\newblock In {\em Proceedings of the Twenty-fourth Annual ACM-SIAM Symposium on
  Discrete Algorithms}, SODA '13, pages 101--107, Philadelphia, PA, USA, 2013.
  Society for Industrial and Applied Mathematics.

\bibitem{EdenFFK18}
Alon Eden, Michal Feldman, Amos Fiat, and Kineret Segal.
\newblock An economic-based analysis of {RANKING} for online bipartite
  matching.
\newblock {\em CoRR}, abs/1804.06637, 2018.

\bibitem{Ehsani2017}
Soheil Ehsani, MohammadTaghi Hajiaghayi, Thomas Kesselheim, and Sahil Singla.
\newblock Prophet secretary for combinatorial auctions and matroids.
\newblock In {\em Proceedings of the Twenty-Ninth Annual ACM-SIAM Symposium on
  Discrete Algorithms}, SODA ’18, page 700–714, USA, 2018. Society for
  Industrial and Applied Mathematics.

\bibitem{EsfandiariHLM19}
Hossein Esfandiari, Mohammad~Taghi Hajiaghayi, Brendan Lucier, and Michael
  Mitzenmacher.
\newblock Online pandora's boxes and bandits.
\newblock In {\em The Thirty-Third {AAAI} Conference on Artificial
  Intelligence, {AAAI} 2019, The Thirty-First Innovative Applications of
  Artificial Intelligence Conference, {IAAI} 2019, The Ninth {AAAI} Symposium
  on Educational Advances in Artificial Intelligence, {EAAI} 2019, Honolulu,
  Hawaii, USA, January 27 - February 1, 2019}, pages 1885--1892, 2019.

\bibitem{Ezra_2020}
Tomer Ezra, Michal Feldman, Nick Gravin, and Zhihao~Gavin Tang.
\newblock Online stochastic max-weight matching: Prophet inequality for vertex
  and edge arrival models.
\newblock In {\em Proceedings of the 21st ACM Conference on Economics and
  Computation}, EC '20, page 769–787, New York, NY, USA, 2020. Association
  for Computing Machinery.

\bibitem{FeldmanMMM09}
Jon Feldman, Aranyak Mehta, Vahab~S. Mirrokni, and S.~Muthukrishnan.
\newblock Online stochastic matching: Beating 1-1/e.
\newblock In {\em 50th Annual {IEEE} Symposium on Foundations of Computer
  Science, {FOCS} 2009, October 25-27, 2009, Atlanta, Georgia, {USA}}, pages
  117--126, 2009.

\bibitem{Feldman_2016}
Moran Feldman, Ola Svensson, and Rico Zenklusen.
\newblock Online contention resolution schemes.
\newblock In Robert Krauthgamer, editor, {\em Proceedings of the Twenty-Seventh
  Annual {ACM-SIAM} Symposium on Discrete Algorithms, {SODA} 2016, Arlington,
  VA, USA, January 10-12, 2016}, pages 1014--1033. {SIAM}, 2016.

\bibitem{Gamlath2019}
Buddhima Gamlath, Sagar Kale, and Ola Svensson.
\newblock Beating greedy for stochastic bipartite matching.
\newblock In {\em Proceedings of the Thirtieth Annual ACM-SIAM Symposium on
  Discrete Algorithms}, SODA ’19, page 2841–2854, USA, 2019. Society for
  Industrial and Applied Mathematics.

\bibitem{GartnerM}
Bernd G{\"{a}}rtner and Jir{\'{\i}} Matousek.
\newblock {\em Understanding and using linear programming}.
\newblock Universitext. Springer, 2007.

\bibitem{Goyal2020}
Vineet Goyal and Rajan Udwani.
\newblock Online matching with stochastic rewards: Optimal competitive ratio
  via path based formulation.
\newblock In P{\'{e}}ter Bir{\'{o}}, Jason Hartline, Michael Ostrovsky, and
  Ariel~D. Procaccia, editors, {\em {EC} '20: The 21st {ACM} Conference on
  Economics and Computation, Virtual Event, Hungary, July 13-17, 2020}, page
  791. {ACM}, 2020.

\bibitem{GuptaN13}
Anupam Gupta and Viswanath Nagarajan.
\newblock A stochastic probing problem with applications.
\newblock In Michel~X. Goemans and Jos{\'{e}}~R. Correa, editors, {\em Integer
  Programming and Combinatorial Optimization - 16th International Conference,
  {IPCO} 2013, Valpara{\'{\i}}so, Chile, March 18-20, 2013. Proceedings},
  volume 7801 of {\em Lecture Notes in Computer Science}, pages 205--216.
  Springer, 2013.

\bibitem{GuptaNS16}
Anupam Gupta, Viswanath Nagarajan, and Sahil Singla.
\newblock Algorithms and adaptivity gaps for stochastic probing.
\newblock In Robert Krauthgamer, editor, {\em Proceedings of the Twenty-Seventh
  Annual {ACM-SIAM} Symposium on Discrete Algorithms, {SODA} 2016, Arlington,
  VA, USA, January 10-12, 2016}, pages 1731--1747. {SIAM}, 2016.

\bibitem{GuptaNS17}
Anupam Gupta, Viswanath Nagarajan, and Sahil Singla.
\newblock Adaptivity gaps for stochastic probing: Submodular and {XOS}
  functions.
\newblock In Philip~N. Klein, editor, {\em Proceedings of the Twenty-Eighth
  Annual {ACM-SIAM} Symposium on Discrete Algorithms, {SODA} 2017, Barcelona,
  Spain, Hotel Porta Fira, January 16-19}, pages 1688--1702. {SIAM}, 2017.

\bibitem{HajiaghayiKS07}
Mohammad~Taghi Hajiaghayi, Robert~D. Kleinberg, and Tuomas Sandholm.
\newblock Automated online mechanism design and prophet inequalities.
\newblock In {\em Proceedings of the Twenty-Second {AAAI} Conference on
  Artificial Intelligence, July 22-26, 2007, Vancouver, British Columbia,
  Canada}, pages 58--65. {AAAI} Press, 2007.

\bibitem{huang2018online}
Zhiyi Huang, Zhihao Tang, Xiaowei Wu, and Yuhao Zhang.
\newblock Online vertex-weighted bipartite matching: Beating 1-1/e with random
  arrivals.
\newblock {\em ACM Transactions on Algorithms}, 15, 04 2018.

\bibitem{huang2020online}
Zhiyi Huang and Qiankun Zhang.
\newblock Online primal dual meets online matching with stochastic rewards:
  Configuration lp to the rescue.
\newblock In {\em Proceedings of the 52nd Annual ACM SIGACT Symposium on Theory
  of Computing}, STOC 2020, page 1153–1164, New York, NY, USA, 2020.
  Association for Computing Machinery.

\bibitem{KarandeMT11}
Chinmay Karande, Aranyak Mehta, and Pushkar Tripathi.
\newblock Online bipartite matching with unknown distributions.
\newblock In {\em Proceedings of the 43rd {ACM} Symposium on Theory of
  Computing, {STOC} 2011, San Jose, CA, USA, 6-8 June 2011}, pages 587--596,
  2011.

\bibitem{KarpVV90}
Richard~M. Karp, Umesh~V. Vazirani, and Vijay~V. Vazirani.
\newblock An optimal algorithm for on-line bipartite matching.
\newblock In {\em Proceedings of the 22nd Annual {ACM} Symposium on Theory of
  Computing, May 13-17, 1990, Baltimore, Maryland, {USA}}, pages 352--358,
  1990.

\bibitem{KRTV2013}
Thomas Kesselheim, Klaus Radke, Andreas T{\"o}nnis, and Berthold V{\"o}cking.
\newblock An optimal online algorithm for weighted bipartite matching and
  extensions to combinatorial auctions.
\newblock In Hans~L. Bodlaender and Giuseppe~F. Italiano, editors, {\em
  Algorithms -- ESA 2013}, pages 589--600, Berlin, Heidelberg, 2013. Springer
  Berlin Heidelberg.

\bibitem{KrengelS77}
Ulrich Krengel and Louis Sucheston.
\newblock Semiamarts and finite values.
\newblock {\em Bulletin of the American Mathematical Society}, 83(4):745--747,
  1977.

\bibitem{LavastidaMRX20}
Thomas Lavastida, Benjamin Moseley, R.~Ravi, and Chenyang Xu.
\newblock Learnable and instance-robust predictions for online matching, flows
  and load balancing.
\newblock {\em CoRR}, abs/2011.11743, 2020.

\bibitem{Lee2018}
Euiwoong Lee and Sahil Singla.
\newblock {Optimal Online Contention Resolution Schemes via Ex-Ante Prophet
  Inequalities}.
\newblock In Yossi Azar, Hannah Bast, and Grzegorz Herman, editors, {\em 26th
  Annual European Symposium on Algorithms (ESA 2018)}, volume 112 of {\em
  Leibniz International Proceedings in Informatics (LIPIcs)}, pages
  57:1--57:14, Dagstuhl, Germany, 2018. Schloss Dagstuhl--Leibniz-Zentrum fuer
  Informatik.

\bibitem{Lucier17}
Brendan Lucier.
\newblock An economic view of prophet inequalities.
\newblock {\em SIGecom Exch.}, 16(1):24--47, 2017.

\bibitem{Mahdian2011}
Mohammad Mahdian and Qiqi Yan.
\newblock Online bipartite matching with random arrivals: An approach based on
  strongly factor-revealing lps.
\newblock In {\em Proceedings of the Forty-third Annual ACM Symposium on Theory
  of Computing}, STOC '11, pages 597--606, New York, NY, USA, 2011. ACM.

\bibitem{ManshadiGS12}
Vahideh~H. Manshadi, Shayan~Oveis Gharan, and Amin Saberi.
\newblock Online stochastic matching: Online actions based on offline
  statistics.
\newblock {\em Math. Oper. Res.}, 37(4):559--573, 2012.

\bibitem{Mehta13}
Aranyak Mehta.
\newblock Online matching and ad allocation.
\newblock {\em Foundations and Trends in Theoretical Computer Science},
  8(4):265--368, 2013.

\bibitem{MehtaP12}
Aranyak Mehta and Debmalya Panigrahi.
\newblock Online matching with stochastic rewards.
\newblock In {\em 53rd Annual {IEEE} Symposium on Foundations of Computer
  Science, {FOCS} 2012, New Brunswick, NJ, USA, October 20-23, 2012}, pages
  728--737. {IEEE} Computer Society, 2012.

\bibitem{Purohit2019}
Manish Purohit, Sreenivas Gollapudi, and Manish Raghavan.
\newblock Hiring under uncertainty.
\newblock In Kamalika Chaudhuri and Ruslan Salakhutdinov, editors, {\em
  Proceedings of the 36th International Conference on Machine Learning},
  volume~97 of {\em Proceedings of Machine Learning Research}, pages
  5181--5189. PMLR, 09--15 Jun 2019.

\bibitem{Groetschel}
D.~Seese.
\newblock Groetschel, m., l. lovasz, a. schrijver: Geometric algorithms and
  combinatorial optimization. (algorithms and combinatorics. eds.: R. l.
  graham, b. korte, l. lovasz. vol. 2), springer-verlag 1988, xii, 362 pp., 23
  figs., dm 148,-. isbn 3–540–13624-x.
\newblock {\em Biometrical Journal}, 32(8):930--930, 1990.

\bibitem{Segev2020}
Danny Segev and Sahil Singla.
\newblock Efficient approximation schemes for stochastic probing and prophet
  problems.
\newblock {\em CoRR}, abs/2007.13121, 2020.

\bibitem{Tang2020}
Zhihao~Gavin Tang, Xiaowei Wu, and Yuhao Zhang.
\newblock A simple 1-1/e approximation for oblivious bipartite matching.
\newblock {\em CoRR}, abs/2002.06037, 2020.

\bibitem{TangWZ2020}
Zhihao~Gavin Tang, Xiaowei Wu, and Yuhao Zhang.
\newblock Towards a better understanding of randomized greedy matching.
\newblock In {\em Proccedings of the 52nd Annual {ACM} {SIGACT} Symposium on
  Theory of Computing, {STOC} 2020, Chicago, IL, USA, June 22-26, 2020}, pages
  1097--1110, 2020.

\bibitem{Vondrak_2011}
Jan Vondr\'{a}k, Chandra Chekuri, and Rico Zenklusen.
\newblock Submodular function maximization via the multilinear relaxation and
  contention resolution schemes.
\newblock In {\em Proceedings of the Forty-Third Annual ACM Symposium on Theory
  of Computing}, STOC '11, page 783–792, New York, NY, USA, 2011. Association
  for Computing Machinery.

\bibitem{Weitzman1979}
Martin Weitzman.
\newblock Optimal search for the best alternative.
\newblock {\em Econometrica}, 47:641--654, 1979.

\bibitem{Williamson}
David~P. Williamson and David~B. Shmoys.
\newblock {\em The Design of Approximation Algorithms}.
\newblock Cambridge University Press, USA, 1st edition, 2011.

\end{thebibliography}


\begin{thebibliography}{10}

\bibitem{Adamczyk2017}
Marek Adamczyk, Fabrizio Grandoni, Stefano Leonardi, and Michal Wlodarczyk.
\newblock When the optimum is also blind: a new perspective on universal
  optimization.
\newblock In {\em ICALP}, 2017.

\bibitem{Adamczyk15}
Marek Adamczyk, Fabrizio Grandoni, and Joydeep Mukherjee.
\newblock Improved approximation algorithms for stochastic matching.
\newblock In Nikhil Bansal and Irene Finocchi, editors, {\em Algorithms - {ESA}
  2015 - 23rd Annual European Symposium, Patras, Greece, September 14-16, 2015,
  Proceedings}, volume 9294 of {\em Lecture Notes in Computer Science}, pages
  1--12. Springer, 2015.

\bibitem{adamczyk2018random}
Marek Adamczyk and Micha{\l} W{\l}odarczyk.
\newblock Random order contention resolution schemes.
\newblock In {\em 2018 IEEE 59th Annual Symposium on Foundations of Computer
  Science (FOCS)}, pages 790--801. IEEE, 2018.

\bibitem{alaei_2012}
Saeed Alaei, MohammadTaghi Hajiaghayi, and Vahid Liaghat.
\newblock Online prophet-inequality matching with applications to ad
  allocation.
\newblock In {\em Proceedings of the 13th ACM Conference on Electronic
  Commerce}, EC '12, page 18–35, New York, NY, USA, 2012. Association for
  Computing Machinery.

\bibitem{AzarCK18}
Yossi Azar, Ashish Chiplunkar, and Haim Kaplan.
\newblock Prophet secretary: Surpassing the 1-1/e barrier.
\newblock In {\'{E}}va Tardos, Edith Elkind, and Rakesh Vohra, editors, {\em
  Proceedings of the 2018 {ACM} Conference on Economics and Computation,
  Ithaca, NY, USA, June 18-22, 2018}, pages 303--318. {ACM}, 2018.

\bibitem{BansalGLMNR12}
Nikhil Bansal, Anupam Gupta, Jian Li, Juli{\'{a}}n Mestre, Viswanath Nagarajan,
  and Atri Rudra.
\newblock When {LP} is the cure for your matching woes: Improved bounds for
  stochastic matchings.
\newblock {\em Algorithmica}, 63(4):733--762, 2012.

\bibitem{BavejaBCNSX18}
Alok Baveja, Amit Chavan, Andrei Nikiforov, Aravind Srinivasan, and Pan Xu.
\newblock Improved bounds in stochastic matching and optimization.
\newblock {\em Algorithmica}, 80(11):3225--3252, Nov 2018.

\bibitem{borodin2020}
Allan Borodin, Calum MacRury, and Akash Rakheja.
\newblock Bipartite stochastic matching: Online, random order, and i.i.d.
  models.
\newblock {\em CoRR}, abs/2004.14304, 2020.

\bibitem{borodin2021greedy}
Allan Borodin, Calum MacRury, and Akash Rakheja.
\newblock Greedy approaches to online stochastic matching.
\newblock {\em CoRR}, abs/2008.09260, 2021.

\bibitem{borodin2021secretary}
Allan Borodin, Calum MacRury, and Akash Rakheja.
\newblock Secretary matching meets probing with commitment.
\newblock In {\em Approximation, Randomization, and Combinatorial Optimization.
  Algorithms and Techniques, {APPROX/RANDOM} 2021, August 16-18, 2021, Virtual
  Conference}, 2021.

\bibitem{brubach2021follow}
Brian Brubach, Nathaniel Grammel, Will Ma, and Aravind Srinivasan.
\newblock Follow your star: New frameworks for online stochastic matching with
  known and unknown patience.
\newblock {\em CoRR}, abs/1907.03963, 2021.

\bibitem{brubach2021conf}
Brian Brubach, Nathaniel Grammel, Will Ma, and Aravind Srinivasan.
\newblock Follow your star: New frameworks for online stochastic matching with
  known and unknown patience.
\newblock In Arindam Banerjee and Kenji Fukumizu, editors, {\em Proceedings of
  The 24th International Conference on Artificial Intelligence and Statistics},
  volume 130 of {\em Proceedings of Machine Learning Research}, pages
  2872--2880. PMLR, 13--15 Apr 2021.

\bibitem{brubach2021offline}
Brian Brubach, Nathaniel Grammel, Will Ma, and Aravind Srinivasan.
\newblock Improved guarantees for offline stochastic matching via new ordered
  contention resolution schemes.
\newblock {\em CoRR}, abs/2106.06892, 2021.

\bibitem{Brubach2019}
Brian Brubach, Nathaniel Grammel, and Aravind Srinivasan.
\newblock Vertex-weighted online stochastic matching with patience constraints.
\newblock {\em CoRR}, abs/1907.03963, 2019.

\bibitem{BrubachSSX16}
Brian Brubach, Karthik~Abinav Sankararaman, Aravind Srinivasan, and Pan Xu.
\newblock New algorithms, better bounds, and a novel model for online
  stochastic matching.
\newblock In {\em 24th Annual European Symposium on Algorithms, {ESA} 2016,
  August 22-24, 2016, Aarhus, Denmark}, pages 24:1--24:16, 2016.

\bibitem{BrubachSSX20}
Brian Brubach, Karthik~Abinav Sankararaman, Aravind Srinivasan, and Pan Xu.
\newblock Attenuate locally, win globally: Attenuation-based frameworks for
  online stochastic matching with timeouts.
\newblock {\em Algorithmica}, 82(1):64--87, 2020.

\bibitem{ChawlaHMS10}
Shuchi Chawla, Jason~D. Hartline, David~L. Malec, and Balasubramanian Sivan.
\newblock Multi-parameter mechanism design and sequential posted pricing.
\newblock In Leonard~J. Schulman, editor, {\em Proceedings of the 42nd {ACM}
  Symposium on Theory of Computing, {STOC} 2010, Cambridge, Massachusetts, USA,
  5-8 June 2010}, pages 311--320. {ACM}, 2010.

\bibitem{Chen}
Ning Chen, Nicole Immorlica, Anna~R. Karlin, Mohammad Mahdian, and Atri Rudra.
\newblock Approximating matches made in heaven.
\newblock In {\em Proceedings of the 36th International Colloquium on Automata,
  Languages and Programming: Part I}, ICALP '09, pages 266--278, 2009.

\bibitem{CorreaFPV19}
Jos{\'{e}}~R. Correa, Patricio Foncea, Dana Pizarro, and Victor Verdugo.
\newblock From pricing to prophets, and back!
\newblock {\em Oper. Res. Lett.}, 47(1):25--29, 2019.

\bibitem{CorreaSZ19}
Jos{\'{e}}~R. Correa, Raimundo Saona, and Bruno Ziliotto.
\newblock Prophet secretary through blind strategies.
\newblock In {\em Proceedings of the Thirtieth Annual {ACM-SIAM} Symposium on
  Discrete Algorithms, {SODA} 2019, San Diego, California, USA, January 6-9,
  2019}, pages 1946--1961, 2019.

\bibitem{costello2012matching}
Kevin~P. Costello, Prasad Tetali, and Pushkar Tripathi.
\newblock Stochastic matching with commitment.
\newblock In Artur Czumaj, Kurt Mehlhorn, Andrew Pitts, and Roger Wattenhofer,
  editors, {\em Automata, Languages, and Programming}, pages 822--833, Berlin,
  Heidelberg, 2012. Springer Berlin Heidelberg.

\bibitem{EdenFFK18}
Alon Eden, Michal Feldman, Amos Fiat, and Kineret Segal.
\newblock An economic-based analysis of {RANKING} for online bipartite
  matching.
\newblock {\em CoRR}, abs/1804.06637, 2018.

\bibitem{Ehsani2017}
Soheil Ehsani, MohammadTaghi Hajiaghayi, Thomas Kesselheim, and Sahil Singla.
\newblock Prophet secretary for combinatorial auctions and matroids.
\newblock In {\em Proceedings of the Twenty-Ninth Annual ACM-SIAM Symposium on
  Discrete Algorithms}, SODA ’18, page 700–714, USA, 2018. Society for
  Industrial and Applied Mathematics.

\bibitem{Ezra_2020}
Tomer Ezra, Michal Feldman, Nick Gravin, and Zhihao~Gavin Tang.
\newblock Online stochastic max-weight matching: Prophet inequality for vertex
  and edge arrival models.
\newblock In {\em Proceedings of the 21st ACM Conference on Economics and
  Computation}, EC '20, page 769–787, New York, NY, USA, 2020. Association
  for Computing Machinery.

\bibitem{Feldman_2016}
Moran Feldman, Ola Svensson, and Rico Zenklusen.
\newblock Online contention resolution schemes.
\newblock In Robert Krauthgamer, editor, {\em Proceedings of the Twenty-Seventh
  Annual {ACM-SIAM} Symposium on Discrete Algorithms, {SODA} 2016, Arlington,
  VA, USA, January 10-12, 2016}, pages 1014--1033. {SIAM}, 2016.

\bibitem{Gamlath2019}
Buddhima Gamlath, Sagar Kale, and Ola Svensson.
\newblock Beating greedy for stochastic bipartite matching.
\newblock In {\em Proceedings of the Thirtieth Annual ACM-SIAM Symposium on
  Discrete Algorithms}, SODA ’19, page 2841–2854, USA, 2019. Society for
  Industrial and Applied Mathematics.

\bibitem{GandhiGKSP06}
Rajiv Gandhi, Samir Khuller, Srinivasan Parthasarathy, and Aravind Srinivasan.
\newblock Dependent rounding and its applications to approximation algorithms.
\newblock {\em J. ACM}, 53(3):324–360, May 2006.

\bibitem{GartnerM}
Bernd G{\"{a}}rtner and Jir{\'{\i}} Matousek.
\newblock {\em Understanding and using linear programming}.
\newblock Universitext. Springer, 2007.

\bibitem{GuptaNS16}
Anupam Gupta, Viswanath Nagarajan, and Sahil Singla.
\newblock Algorithms and adaptivity gaps for stochastic probing.
\newblock In Robert Krauthgamer, editor, {\em Proceedings of the Twenty-Seventh
  Annual {ACM-SIAM} Symposium on Discrete Algorithms, {SODA} 2016, Arlington,
  VA, USA, January 10-12, 2016}, pages 1731--1747. {SIAM}, 2016.

\bibitem{HajiaghayiKS07}
Mohammad~Taghi Hajiaghayi, Robert~D. Kleinberg, and Tuomas Sandholm.
\newblock Automated online mechanism design and prophet inequalities.
\newblock In {\em Proceedings of the Twenty-Second {AAAI} Conference on
  Artificial Intelligence, July 22-26, 2007, Vancouver, British Columbia,
  Canada}, pages 58--65. {AAAI} Press, 2007.

\bibitem{KarpVV90}
Richard~M. Karp, Umesh~V. Vazirani, and Vijay~V. Vazirani.
\newblock An optimal algorithm for on-line bipartite matching.
\newblock In {\em Proceedings of the 22nd Annual {ACM} Symposium on Theory of
  Computing, May 13-17, 1990, Baltimore, Maryland, {USA}}, pages 352--358,
  1990.

\bibitem{LavastidaMRX20}
Thomas Lavastida, Benjamin Moseley, R.~Ravi, and Chenyang Xu.
\newblock Learnable and instance-robust predictions for online matching, flows
  and load balancing.
\newblock {\em CoRR}, abs/2011.11743, 2020.

\bibitem{Lee2018}
Euiwoong Lee and Sahil Singla.
\newblock {Optimal Online Contention Resolution Schemes via Ex-Ante Prophet
  Inequalities}.
\newblock In Yossi Azar, Hannah Bast, and Grzegorz Herman, editors, {\em 26th
  Annual European Symposium on Algorithms (ESA 2018)}, volume 112 of {\em
  Leibniz International Proceedings in Informatics (LIPIcs)}, pages
  57:1--57:14, Dagstuhl, Germany, 2018. Schloss Dagstuhl--Leibniz-Zentrum fuer
  Informatik.

\bibitem{ManshadiGS12}
Vahideh~H. Manshadi, Shayan~Oveis Gharan, and Amin Saberi.
\newblock Online stochastic matching: Online actions based on offline
  statistics.
\newblock {\em Math. Oper. Res.}, 37(4):559--573, 2012.

\bibitem{Purohit2019}
Manish Purohit, Sreenivas Gollapudi, and Manish Raghavan.
\newblock Hiring under uncertainty.
\newblock In Kamalika Chaudhuri and Ruslan Salakhutdinov, editors, {\em
  Proceedings of the 36th International Conference on Machine Learning},
  volume~97 of {\em Proceedings of Machine Learning Research}, pages
  5181--5189. PMLR, 09--15 Jun 2019.

\bibitem{Groetschel}
D.~Seese.
\newblock Groetschel, m., l. lovasz, a. schrijver: Geometric algorithms and
  combinatorial optimization. (algorithms and combinatorics. eds.: R. l.
  graham, b. korte, l. lovasz. vol. 2), springer-verlag 1988, xii, 362 pp., 23
  figs., dm 148,-. isbn 3–540–13624-x.
\newblock {\em Biometrical Journal}, 32(8):930--930, 1990.

\bibitem{Segev2020}
Danny Segev and Sahil Singla.
\newblock Efficient approximation schemes for stochastic probing and prophet
  problems.
\newblock In {\em Proceedings of the 22nd ACM Conference on Economics and
  Computation}, EC '21, page 793–794, New York, NY, USA, 2021. Association
  for Computing Machinery.

\bibitem{Vondrak_2011}
Jan Vondr\'{a}k, Chandra Chekuri, and Rico Zenklusen.
\newblock Submodular function maximization via the multilinear relaxation and
  contention resolution schemes.
\newblock In {\em Proceedings of the Forty-Third Annual ACM Symposium on Theory
  of Computing}, STOC '11, page 783–792, New York, NY, USA, 2011. Association
  for Computing Machinery.

\bibitem{Williamson}
David~P. Williamson and David~B. Shmoys.
\newblock {\em The Design of Approximation Algorithms}.
\newblock Cambridge University Press, USA, 1st edition, 2011.

\end{thebibliography}
